\documentclass{article}
\usepackage[a4paper,top=3cm,bottom=3cm,left=3cm,right=3cm,marginparwidth=3cm]{geometry}
\usepackage{amssymb,amsthm,mathtools}
\usepackage{hyperref}
\hypersetup{
    colorlinks=false,
    linkcolor=blue,
    filecolor=magenta,      
    urlcolor=magenta,
}
 \usepackage{setspace}

\newcommand{\N}{\ensuremath{\mathbb{N}} }

\newcommand{\fbins}{\ensuremath{\{0,1\}^*}}
\newcommand{\bin}{\ensuremath{\{0,1\}}}
\newcommand{\infbins}{\ensuremath{\{0,1\}^{\omega}}}
\newcommand{\ILPDC}{\ensuremath{\mathrm{ILPDC}}}
\newcommand{\ILFST}{\ensuremath{\mathrm{ILFST}}}
\newcommand{\pref}{\ensuremath{\textrm{pref}}}

\newcommand{\FScomp}[1]{\ensuremath{D_{\mathrm{FS}}^{#1}}}
\newcommand{\Fcomp}[1]{\ensuremath{D_{\mathrm{F}}^{#1}}}
\newcommand{\PBcomp}[1]{\ensuremath{D_{\mathrm{PB}}^{#1}}}
\newcommand{\PB}{\ensuremath{\mathrm{PB}}}
\newcommand{\push}{\ensuremath{\textrm{push}}}

\newcommand{\ILUPDC}{\ensuremath{\mathrm{ILUPDC}}}
\newcommand{\pop}{\ensuremath{\textrm{pop}}}
\newcommand{\FSTsize}[1]{\text{FST}^{\leq #1}}
\newcommand{\Fsize}[1]{\text{F}^{\leq #1}}
\newcommand{\Ft}{\ensuremath{\textrm{F}}}

\newcommand{\thh}{\ensuremath{\textrm{th}}}
\newcommand{\PBsize}[1]{\ensuremath{\mathrm{PB}^{\leq #1}}}

\theoremstyle{plain}
\newtheorem{theorem}{Theorem}[section]
\newtheorem{corollary}[theorem]{Corollary}
\newtheorem{lemma}[theorem]{Lemma}
\theoremstyle{definition}
\newtheorem{definition}[theorem]{Definition}

\newtheorem{remark}[theorem]{Remark}
\newtheorem{claim}[theorem]{Claim}

\begin{document}

\title{Pebble-Depth}

\author{Liam Jordon\thanks{Supported by a postgraduate scholarship from the Irish Research Council.}\\
\textrm{liam.jordon@mu.ie} \\
\and 
Philippe Moser \\
\textrm{philippe.moser@mu.ie}}

\date{%
    Dept. of Computer Science, Maynooth University, Maynooth, Co. Kildare, Ireland.\\[2ex]%
}

\maketitle

 \begin{abstract}
        In this paper we introduce a new formulation of Bennett's logical depth based on pebble transducers. This notion is defined based on the difference between the minimal length descriptional complexity of prefixes of infinite sequences from the perspective of finite-state transducers and pebble transducers. Our notion of pebble-depth satisfies the three fundamental properties of depth: i.e. easy sequences and random sequences are not deep, and the existence of a slow growth law type result. We also compare pebble-depth to other depth notions based on finite-state transducers, pushdown compressors and the Lempel-Ziv $78$ compression algorithm. We first demonstrate that there exists a normal pebble-deep sequence even though there is no normal finite-state-deep sequence. We then show that there exists a sequence which has pebble-depth level of roughly $1/2$ and Lempel-Ziv-depth level of roughly $0$. Finally we show the existence of a sequence which has a pebble-depth level of roughly $1$ and a pushdown-depth level of roughly $1/2$.
    \end{abstract}


\section{Introduction}

In 1988 Charles Bennett introduced a new method to measure the complexity of a piece of data \cite{b:bennett88}. He called this tool \textit{logical depth}. Intuitively, deep structures can be thought of as structures which contain complex patterns that are hard to find. Given more time and resources, an algorithm could identify patterns in the structure to exploit them. Non-deep structures are referred to as being \emph{shallow}. Random structures are considered shallow as they contain no patterns to identify. Simple structures are considered shallow as while they contain patterns, they are too easy to spot.

Bennett's original notion was based on uncomputable Kolmogorov complexity and has been shown to interact nicely with several aspects of computability theory \cite{DBLP:journals/tcs/DowneyMN17,DBLP:journals/dmtcs/MoserS17}. Recently however, researchers have studied more feasible notions at lower complexity levels. These include computable notions \cite{DBLP:journals/iandc/LathropL99}, notions based on polynomial time computations \cite{b.antunes.depth.journal,DBLP:journals/tcs/Moser13,DBLP:journals/iandc/Moser20}, and notions based on classes of transducers \cite{DotyM07,JordonM20}.

Many of these notions are based on minimal descriptional complexity, i.e. the ratio of the length of input and output of a function. For notions based on classes of transducers (finite-state and pushdown) previously studied, this ratio has been linear \cite{DotyM07,JordonM20}. In this paper we examine the minimal descriptional complexity of a class of transducers known as \textit{pebble} transducers which has polynomial size output.

 For $k \in \N$, a $k$-pebble transducer is a two-way finite-state transducer with the additional capacity to mark $k$ squares of its tape with its \textit{pebbles}.  When two or more pebbles are used, we exclusively examine pebble transducers whose pebbles follow a stack-like discipline in the sense that the pebbles are ordered and a pebble's position on the tape can only be altered (lifted from or dropped onto a square) if all lower ranked pebbles are currently on the tape and all higher ranked pebbles are not currently on the tape. Building on work studying $0$-pebble automata \cite{DBLP:journals/ibmrd/RabinS59,DBLP:journals/ibmrd/Shepherdson59} and $1$-pebble automata \cite{DBLP:conf/focs/BlumH67}, Globerman and Harel showed that this restricted class of pebble automata recognise exactly the regular languages \cite{DBLP:journals/tcs/GlobermanH96}. They were first examined as transducers for trees by Milo et al. in \cite{DBLP:journals/jcss/MiloSV03}. Further study of pebble-transducers for strings by Engelfriet and Maneth can be found in \cite{DBLP:conf/mfcs/EngelfrietM02}. While Globerman and Harel's result shows that pebble acceptors are equivalent to finite-state acceptors, pebble transducers are much more powerful in the following sense: While two-way finite-state transducers have output $O(n)$ on inputs of size $n$, a $k$-pebble transducer has output of size $O(n^{k+1})$. As such, the class of functions computed by pebble transducers has been referred to as the class of \textit{polyregular} functions and has been shown to have several equivalent characterizations \cite{MikoPoly18,DBLP:conf/icalp/BojanczykKL19}. Specifically, a recent result by Lhote demonstrates that a polyregular function has output of size $O(n^{k+1})$ if and only if the function can be performed by a $k$-pebble transducer \cite{DBLP:conf/lics/Lhote20}.

Depth notions are defined via two families of string to string functions called \textit{observers} (e.g. lossless compressors) $T$ and $T'$ where $T'$ is more powerful than $T$. We say an infinite sequence $S$ is $(T,T')$-deep if for every observer $G$ of type $T$, there exists an observer $G'$ of type $T'$ such that on almost every prefix (or infinitely many prefixes) of $S$, $G'$ performs better (e.g. achieves better compression) than $G$ on the prefix by at least $\alpha n$ bits where $\alpha \in (0,1)$ and $n$ is the length of the prefix. We refer to $\alpha$ as the $(T,T')$-depth level of $S$. Bennett's original notion was based on time bounded Kolmogorov complexity $H^t$ and ordinary Kolmogorov complexity $H$, i.e. $(H^t,H)$-depth. 

Inspired by this, we define pebble-depth (PB-depth) to be based on the difference between the minimal descriptional complexity of strings when restricted to finite-state transducers and pebble transducers, i.e. $(\textrm{FS},\textrm{PB})$-depth. We propose that finite-state transducers are a good candidate for the weaker classes of observers as finite-state transducers can be viewed as $0$-pebble transducers with the added restriction that the tape head can only move in one direction. Finite-state transducers have previously been used to define notions of depth in \cite{DotyM07,JordonM20}. We prove our new pebble notion satisfies the three basic properties of depth that are generally considered fundamental. The first two are that random sequences are shallow (for the appropriate randomness notion) and computable sequences are shallow (for the appropriate randomness notion). The third fundamental property featured in Bennett's original notion is called the \textit{slow growth law} which states that no fast simple process can create a deep sequence. Therefore, the construction of a deep sequence must be in someway difficult. We demonstrate that our PB-depth also satisfies a slow growth law.

We compare PB-depth with finite-state-depth (FS-depth) which is based on finite-state transducers \cite{DotyM07}, Lempel-Ziv-depth (LZ-depth) \cite{JordonMPDLZ} based on the lossless compression algorithm Lempel-Ziv $78$ \cite{DBLP:LZ78}, and with pushdown-depth (PD-depth) \cite{JordonMPDLZ} which is based on information lossless pushdown compressors. We first demonstrate that unlike FS-depth where all normal sequences are shallow, there exists a normal PB-deep sequence. We demonstrate a difference with LZ-depth by building a sequence which has a PB-depth level of approximately $1/2$ and an LZ-depth level close to $0$. We also offer preliminary results comparing PB-depth with PD-depth by showing the existence of a sequence which has a PB-depth level of close to $1$ and a PD-depth level of close to $1/2$.

\section{Preliminaries}

We write $\N$ to denote the set of all non-negative integers. All logarithms are taken in base $2$. A \emph{string} is an element of $\fbins$. We use $|x|$ to denote the length of string $x$. For $n \in \N$, $\bin^n$ denotes the set of strings of length $n$. We use $\bin^{+}$ to denote the set of strings of length at least $1$. A \emph{sequence} is an element of $\infbins$. Given strings $x,y$ and a sequence $S$, $xy$ and $xS$ denote the concatenation of $x$ with $y$ and $x$ with $S$
 respectively. For a string $x$ and $n\in \N$, $x^n$ denotes the string of $n$ copies of $x$ concatenated together while $x^{\omega}$ denotes the sequence formed from infinite concatenations of $x$ with itself. For a string $x$ and sequence $S$, for positive integers $i,j$ with $i\leq j$, $x[i..j]$ and $S[i..j]$ represent the \emph{substring} of $x$ and $S$ composed of their respective $(i+1)^{\thh}$ through $(j+1)^{\thh}$ bits. If $j < i$, then $x[i.. j]=S[i .. j] = \lambda$, where $\lambda$ is the \emph{empty} string. $x[i]$ and $S[i]$ represent the $i^{\thh}$ bit of $x$ and $S$ respectively. For a string $v = xyz$ we say $x$ is a \emph{prefix} of $v$, $y$ is a substring of $v$, and that $z$ is a suffix of $v$. We write $x \sqsubseteq v$ to denote that $x$ is a prefix of $v$ and $x \sqsubset v$ if $x$ is a prefix of $v$ but $x \neq v$. For a string or sequence $S$ and $n \in \N$, $S \upharpoonright n$ denotes $S[0 .. n-1]$, the prefix of length $n$ of $S$. 
 
 For a string $x = x_1x_2 \ldots x_n$, $d(x)$ denotes $x$ with every bit doubled, i.e.  $d(x)= x_1x_1 \ldots x_nx_n$. $x^{-1}$ denotes the reverse of $x$, i.e. $x^{-1} = x_nx_{n-1} \ldots x_2x_1$. For $k \in \N$, $pow_k(x)$ denotes the string $x^{|x|^k}$, i.e. the string formed by concatenating $x$ with itself $|x|^k$ times. $\pref(x)$ denotes the string $x_1x_1x_2x_1x_2x_3\ldots x$, i.e. the string formed by concatenating all prefixes of $x$ in order of length.

We write $K(x)$ to represent the plain Kolmogorov complexity of $x$. That is, for a fixed universal Turing machine $U$, $K_U(x) = \min \{ |y| : y \in \fbins, U(y) = x \}.$ Here $y$ is the shortest input to $U$ that results in the output of $x$. The value $K_U(x)$ does not depend on the choice of universal machine up to an additive constant, therefore we drop the $U$ from the notation. Other authors commonly use $C$ to denote plain Kolmogorov complexity (see \cite{nies:book}), however we reserve $C$ to denote compressors. Note that for all $n \in \N$, there exists a string $x \in \bin^n$ such that $K(x) \geq |x|$ by a simple counting argument. The \emph{prefix-free} variation of Kolmogorov Complexity, denoted by $H$, is similarly defined when Turing machines we examine are restricted to having a prefix-free domain.

We use Borel normality \cite{borelNormal} to examine the properties of some sequences. A sequence $S$ is said to be \emph{normal} if for all $x\in \fbins$, $x$ occurs with asymptotic frequency $2^{-|x|}$ as a substring in $S$.

\subsection{Pebble Transducers}

A $k$-pebble transducer is two-way finite-state transducer which also has $k$-pebbles labelled $1,\ldots,k$. Initially, the transducer has no pebbles on its input tape, however during its computation the transducer can drop pebbles onto and pick up pebbles from squares of the tape. At each stage of computation the transducer knows which pebbles are on the square under its head and it can choose to drop a new pebble on that square, lift the topmost pebble from that square, or move to a different square. However, the way in which the transducer can drop and lift pebbles is restricted to act in a stack like fashion.

We use the pebble transducers with this restriction as these transducers have very nice properties. For instance, in Theorem $2$ of \cite{DBLP:conf/mfcs/EngelfrietM02} it was shown that this class of deterministic pebble transducers is closed under composition. We use this property in the proof of Theorem \ref{PB: Thm: SGL}. These transducers also have other nice properties such as that all non-deterministic pebble transducers which compute a partial function can in fact be computed by a deterministic pebble transducer \cite{DBLP:journals/acta/Engelfriet15}.

\begin{definition}
A \textit{pebble transducer (PB)} is a $6$-tuple $T = (Q,q_0,F,k,\delta,\nu)$ where
\begin{enumerate}
    \item $Q$ is a non-empty, finite set of \emph{states},
    \item $q_0 \in Q$ is the \textit{start state},
    \item $F \subseteq Q $ is the set of \textit{final states},
    \item $k$ is the number of \emph{pebbles} allowed to be placed on the tape,
    \item $\delta:(Q - F) \times (\bin \bigcup \{\dashv,\vdash\}) \times \bin^k \rightarrow Q \times \{+1,-1,\push,\pop\}$ is the \textit{transition function},
    \item $\nu:(Q-F) \times (\bin \bigcup \{\dashv,\vdash\}) \times \bin^k \rightarrow \fbins$ is the \textit{output function}.
\end{enumerate}
\label{PB: Def: Pebble Transducer}
\end{definition}

A PB with $k$ pebbles is referred to as a $k$-pebble transducer. On input $x \in \fbins$, the input tape contains $\dashv x \vdash$, where $\dashv$ and $\vdash$ are the left and right end markers of the tape respectively. The tape squares are numbered $0,1,\ldots,|x|,|x|+1$. A \emph{configuration} of $T$ is a $4$-tuple $(q,i,\sigma,w)$ where $q \in Q$ is the current state of $T$, $0 \leq i \leq |x|+1$ is the current position of the head, $\sigma \in \{\bot,0,\ldots ,|x|+1\}^k$ is a tuple indicating the location of the pebbles and $w \in \fbins$ is what $T$ has outputted so far. That is, $\sigma[m-1] = j$ means that pebble $m$ is on square $j$, and $\sigma[m-1] = \bot$ means pebble $m$ is not currently on the tape. Hence by the stack nature of $T$, if only $l$ pebbles are currently placed on the tape then $\sigma[l] = \cdots = \sigma[k-1] = \bot$.

If $T$ is in configuration $(q,i,\sigma,w)$ with $a$ being the symbol on square $i$ of the tape, then $T$'s transition and output functions $\delta$ and $\nu$ take the input $(q,a,b)$ where for $0\leq j \leq k-1$, $b[j] = 1 \iff \sigma[j] = i$.

 If there are $l$ pebbles on $T$'s tape, $\delta(q,a,b)=(q',d)$ means that $T$ moves from state $q$ to state $q'$ and performs action $d$, where $+1$ means move one square right, $-1$ means move one square left, $\push$ means place pebble $l+1$ onto the current square and $\pop$ means remove pebble $l$ from the current square. These four types transitions are undefined if performing $d$ results in an impossible action, i.e. if $d = +1$ when $a = \,  \vdash$, $d = -1$ when $a = \, \dashv$, $d = \push$ when all pebbles are currently on the tape, and $d = \pop$ when pebble $l$ is not on the current square of the configuration respectively. Following a transition, $T$ enters the \emph{successor} configuration $(q',i',\sigma',w\cdot \nu(q,a,b))$, where $q',i'$ and $\sigma'$ reflect the new state, head position and place of the pebbles based on the result of $\delta(q,a,b)$. Specifically $q' = \delta_Q(q,a,b),$ $i' \in \{i-1,i,1+1\}$ depending on whether the instruction was to move left, push or pop, or move right respectively and
         \begin{equation}
         \sigma' = \begin{cases}
            \sigma \textrm{ if the instruction was $+1$ or $-1$}, \\
            \sigma[0..l-1]i\bot^{k-l-1} \textrm{ if the instruction   was to push pebble $l+1$},\\
            \sigma[0..l-2]\bot^{k-l+1} \textrm{ if the instruction  was to pop pebble $l$}.
        
                \end{cases}
                \label{PB: Sec: PBT: EQ1}
\end{equation}
    
We say that $T$ on input $x$ outputs $w$, i.e. $T(x) = w,$ if starting in configuration $(q_0,0,\bot^k,\lambda)$, there is a finite sequence of successor configurations of $T$ ending with $(q,i,\sigma,w)$, where $q \in F$. We require the use of final states to define the output as we do not wish to consider cases where $T$ finds itself in a loop and outputs an infinite sequence, i.e. $T(x) = zy^{\omega}$, for some $z,y\in \fbins$ where $|y| \geq 1$. We write \textrm{PB} to denote the set of deterministic pebble transducers. We use PB$_k \subset$ PB to denote set of deterministic $k$-pebble transducers. Note then that each $M \in \textrm{PB}$ computes a partial function from $\fbins$ to $\fbins$.

Using constructions of \cite{DBLP:journals/ita/GeffertI10}, Engelfriet showed that the set of functions computed by PBs is closed under composition \cite{DBLP:journals/acta/Engelfriet15}. This property is used in the proof of a slow growth law in Theorem \ref{PB: Thm: SGL}.

 \begin{theorem}[\cite{DBLP:journals/acta/Engelfriet15}]
   Let $r,m \geq 0$. Let $R \in $ PB$_r$ and $M \in $PB$_m$. Then there exists $T \in$ PB$_{rm + r + m}$ such that for all $x \in \fbins$, $T(x) = R(M(x)).$
    \label{PB: Thm: PB Composition}
\end{theorem}

\subsection{Finite-State Transducers}

We use the standard finite-state transducer model.

\begin{definition}
 A \emph{finite-state transducer (FST)} is a $4$-tuple $T = (Q,q_0,\delta,\nu)$, where
\begin{itemize}
    \item $Q$ is a non-empty, finite set of \emph{states},
    \item $q_0 \in Q$ is the \emph{initial state},
    %\item $A$ is the finite input alphabet,
    %\item $B$ is the finite output alphabet,
    \item $\delta : Q\times \bin \rightarrow Q$ is the \emph{transition function}, 
    \item $\nu : Q\times \bin \rightarrow \fbins$ is the \emph{output function}.
\end{itemize}
\label{FS: Def: FST}
\end{definition}

For all $x \in \fbins$ and $b \in \bin$, the \emph{extended transition function} $\widehat{\delta}:\fbins \rightarrow Q$ is defined by the recursion $\widehat{\delta}(\lambda) = q_0$ and $\widehat{\delta}(xb) = \delta(\widehat{\delta}(x),b).$ For $x \in \fbins$, the output of $T$ on $x$ is the string $T(x)$ defined by the recursion $T(\lambda) = \lambda$, and $T(xb) = T(x)\nu(\widehat{\delta}(x),b)$. We require the class of \textit{information lossless} finite-state transducers to later demonstrate a slow growth law.

\begin{definition}
An FST $T$ is \emph{information lossless (IL)} if for all $x \in \fbins$, the function $x \mapsto (T(x), \widehat{\delta}(x))$ is injective. 
\label{FS: Def: ILFST}
\end{definition}

In other words, an FST $T$ is IL if the output and final state of $T$ on input $x$ uniquely identify $x$. We call an FST that is IL an ILFST. By the identity FST, we mean the ILFST $I_{\mathrm{FS}}$ such that on every input $x$, $I_{\mathrm{FS}}(x) = x.$ We write (IL)FST to denote the set of all (IL)FSTs. We note that occasionally we call ILFSTs \textit{finite-state compressors} to emphasise when we view the ILFSTs as compressors as opposed to decompressors.

We require the concept of \textit{(information lossless) finite-state computable} functions to demonstrate our slow growth also.

\begin{definition}
A function $f: \infbins \rightarrow \infbins$ is said to be  \emph{(information lossless) finite-state computable ((IL)FS computable)} if there is an (IL)FST $T$ such that for all $S \in \infbins$, $\lim\limits_{n \to \infty}|T(S \upharpoonright n)| = \infty$ and for all $n \in \N$, $T(S \upharpoonright n) \sqsubseteq f(S)$.
\label{FS: Def: ILFS Computable}
\end{definition}
Based on the above definition, if $f$ is (IL)FS computable via the (IL)FST $T$, we say that $T(S) = f(S)$.
We often use the following two results \cite{Huff59a,Koha78} that demonstrate that any function computed by an ILFST can be inverted to be approximately computed by another ILFST. 

\begin{theorem}[\cite{Huff59a,Koha78}]
 For all $T \in \ILFST$, there exists $T^{-1} \in \ILFST$ and a constant $c \in \N$ such that for all $x \in \fbins$, $x \upharpoonright (|x| - c) \sqsubseteq T^{-1}(T(x)) \sqsubseteq x$.
 \label{FS: Thm: inverse fst}
\end{theorem}

\subsection{Lempel-Ziv 78 Algotithm}
The Lempel-Ziv 78 algorithm (denoted LZ) \cite{DBLP:LZ78} is a lossless dictionary based compression algorithm. Given an input $x \in \fbins$, LZ parses $x$ into phrases $x = x_1x_2\ldots x_n$ such that each phrase $x_i$ is unique in the parsing, except for possibly the last phrase. Furthermore, for each phrase $x_i$, every prefix of $x_i$ also appears as a phrase in the parsing. That is, if $y \sqsubset x_i$, then $y = x_j$ for some $j<i$. Each phrase is stored in LZ's dictionary. LZ encodes $x$ by encoding each phrase as a pointer to its dictionary containing the longest proper prefix of the phrase along with the final bit of the phrase. Specifically for each phrase $x_i$, $x_i = x_{l(i)}b_i$ for $l(i) < i$ and $b_i \in \{0,1\}.$ Then for $x = x_1x_2\ldots x_n$
$$LZ(x) = c_{l(1)}b_1c_{l(2)}b_2\ldots c_{l(n)}b_n$$
where $c_i$ is a prefix free encoding of the pointer to the $i^{th}$ element of LZ's dictionary, and $x_0 = \lambda$.

\subsection{Pushdown Compressors}

The model of pushdown compressors (PDC) we use to define pushdown depth can be found in \cite{DBLP:journals/mst/MayordomoMP11} where PDCs were referred to as \textit{bounded pushdown compressors}. We use this model  as it allows for feasible run times by bounding the number of times a PDC can pop a bit from its stack without reading an input bit. This prevents the compressor spending an arbitrarily long time altering its stack without reading its input. This model of pushdown compression also has the nice property that it equivalent to a notion of pushdown-dimension based on bounded pushdown gamblers \cite{DBLP:conf/birthday/AlbertMM17}. Similar models where there is not bound on the number of times a bit can be popped off from the stack can be found in \cite{DBLP:journals/tcs/DotyN07}. %Other more restricted forms of pushdown compressors have been studied known based on \textit{visibly pushdown} automata  \cite{DBLP:journals/jacm/AlurM09,DBLP:journals/jcss/FiliotRRST18, DBLP:conf/www/KumarMV07}

The following contains details of the model used. It is taken from \cite{DBLP:journals/mst/MayordomoMP11}.
\begin{definition}
A \emph{pushdown compressor (PDC)} is a 7-tuple $C = (Q, \Gamma, \delta, \nu, q_0, z_0,c)$ where
\begin{enumerate}
    \item $Q$ is a non-empty, finite set of \textit{states},
    %\item $A$ is the finite input alphabet,
    %\item $B$ is the finite output alphabet,
    \item $\Gamma = \{0,1,z_0\}$ is the finite \textit{stack alphabet},
    \item $\delta: Q \times (\bin \cup \{\lambda\}) \times \Gamma \rightarrow Q \times \Gamma^*$ is the \emph{transition function},
    \item $\nu: Q \times (\bin \cup \{\lambda\}) \times \Gamma \rightarrow \fbins$ is the \emph{output function},
    \item $q_0 \in Q$ is the \textit{initial state},
    \item $z_0 \in \Gamma$ is the special \textit{bottom of stack symbol},
    \item $c\in \N$ is an upper bound on the number of $\lambda$-transitions per input bit.
\end{enumerate}
\label{PD: Def: PD}
\end{definition}

We write $\delta_Q$ and $\delta_{\Gamma*}$ to represent the projections of the function $\delta.$ For $z \in \Gamma^+$ the stack of $C$, $z$ is ordered such that $z[0]$ is the topmost symbol of the stack and $z[|z|-1] = z_0.$ $\delta$ is restricted to prevent $z_0$ being popped from the bottom of the stack. That is, for every $q \in Q, \, b \in \bin \cup \{\lambda\}$, either $\delta(q,b,z_0) = \bot$, or $\delta(q,b,z_0) = (q',vz_0)$ where $q' \in Q$ and $v \in \Gamma^*$.

Note that $\delta$ accepts $\lambda$ as a valid input symbol. This means that $C$ has the option to pop the top symbol from its stack and move to another state without reading an input bit. This type of transition is call a $\lambda$\textit{-transition}. In this scenario $\delta(q,\lambda, a) = (q',\lambda)$. To enforce determinism, we ensure that one of the following hold for all $q \in Q$ and $a \in \Gamma$: \begin{itemize}
    \item $\delta(q,\lambda,a) = \bot$, or
    \item $\delta(q,b,a) = \bot$ for all $b \in \bin$.
\end{itemize}
This means that the compressor does not have a choice to read either $0$ or $1$ characters. To prevent an arbitrary number of $\lambda$-transitions occurring at any one time, we restrict $\delta$ such that at most $c$ $\lambda$-transitions can be performed in succession without reading an input bit.

The extended transition function $\widehat{\delta}: Q \times \fbins \times \Gamma^+ \rightarrow Q \times \Gamma^*$ is defined by the usual induction similar to the FST case. $\widehat{\delta}(q_0,w,z_0)$ is abbreviated to $\widehat{\delta}(w)$. The \emph{extended output function} $\widehat{\nu}: Q\times \fbins \times \Gamma^+ \rightarrow Q \times \Gamma^*$ is also defined by the usual induction similar to the FST case also. We omit the sepcifics for space constraints, however, for full details on both, see \cite{DBLP:journals/mst/MayordomoMP11}. The \emph{output} of the PDC $C$ on input $w \in \fbins$ is the string $C(w) = \widehat{\nu}(q_0,w,z_0).$

To make our notion of depth meaningful, we examine the class of \textit{information lossless pushdown compressors}.

\begin{definition}
A PDC $C$ is \textit{information lossless (IL)} if for all $x \in \fbins$, the function $x \mapsto (C(x),\delta_Q(x))$ is injective.
\label{PD: Def: ILPDC}
\end{definition}

In other words, a PDC $C$ is IL if the output and final state of $C$ on input $x$ uniquely identify $x$. We call a PDC that is IL an ILPDC. We write (IL)PDC to denote the set of all (IL)PDCs. By the identity PDC, we mean the ILPDC $I_{\textrm{PD}}$ where on every input $x$, $I_{\textrm{PD}}(x) = x.$

As part of our definition of pushdown depth, we examine ILPDCs whose stack is limited to only containing the symbol $0$ also.

\subsubsection{Unary-stack Pushdown Compressors}

Unary-stack pushdown compressors (UPDCs) are similar to counter compressors as seen in \cite{DBLP:journals/jcss/BecherCH15}. The difference here is that for a UPDC, only a single $0$ can be popped from the stack during a single transtion, while for a counter transducer, an arbitrary number of $0$s can be popped from its stack on a single transition, i.e. its counter can be deducted by an arbitrary amount. However, the UPDC has the ability to pop off $0$s from its stack without reading a symbol via $\lambda$-transitions while the counter compressor cannot. Thus, if a counter compressor decrements its counter by the value of $k$ on a single transition, a UPDC can do the same by performing $k-1$ $\lambda$-transitions in a row before reading performing the transition of the counter compressor and popping off the final $0$.

\begin{definition}
A \emph{unary-stack pushdown compressor (UPDC)} is a $7$-tuple $$C = (Q,\Gamma, \delta, \nu, q_0, z_0, c)$$ where $Q,\delta, \nu, q_0, z_0$ and $c$ are all defined the same as for a PDC in Defintion \ref{PD: Def: PD}, while the \emph{stack alphabet} $\Gamma$ is the set $\{0,z_0\}$.
\label{UPDC def}
\end{definition}

\begin{definition}
A UPDC $C$ is \textit{information lossless (IL)} if for all $x \in \fbins$, the function $x \mapsto (C(x),\delta_Q(x))$ is injective. A UPDC which is IL is referred to as an ILUPDC.
\label{PD: Def: ILUPDC}
\end{definition}

We make the following observation regarding ILUPDCs. Let $C \in \ILUPDC$ and suppose it has been given the input $yx$. After reading the prefix $y$, if $C$'s stack height is large enough such that it never empties on reading the suffix $x$, the actual height of the stack doesn't matter. That is, any reading of $x$ with an arbitrarily large stack which is far enough away from being empty will all have a similar behaviour if starting in the same state. This is because if the stack does not empty, it has little impact on the processing of $x$. We describe this below.

\begin{remark}
Let $C \in \ILUPDC$ and suppose $C$ can perform at most $c$ $\lambda$-transitions in a row. Consider running $C$ on an input of the form $yx$ and let $q$ be the state $C$ ends in after reading $y$. If $C$'s stack has a height above $(c+1)|x|$ after reading $y$, then $C$'s stack can never be fully emptied upon reading $x$. Hence, for $k,k' \geq (c+1)|x|$ with $k\neq k'$ then $C(q,x,0^kz_0) = C(q,x,0^{k'}z_0),$ i.e. $C$ will output the same string regardless of whether the height is $k$ or $k'$. Thus, prior to reading $x$, only knowing whether the stack's height is below $(c+1)|x|$ will have any importance.
\label{PD: Rmk: height discussion}
\end{remark}

 \subsection{\textit{k}-String Complexity}
 
 In \cite{DotyM07}, a notion of finite-state-depth (FS-depth) was introduced which was based on the finite-state minimal descriptional complexity of strings. Further study of finite-state minimal descriptional complexity can be found in \cite{calude2011finite,calude2016finite}. We generalise this idea to any class of transducers $T$ to define our depth notions. In particular, we consider finite-state and pebble-complexity to later define PB-depth.
 
\begin{definition}
Let F be a class of transducers and $D \subseteq \fbins$ be an infinite, computable set of strings. A \emph{binary representation of} $\Ft$\emph{-transducers} $\sigma$ is a computable map $\sigma : D \rightarrow \mathrm{F}$, such that for every transducer $T\in\Ft$, there exists some $x \in D$ such that $\sigma(x)=T$, i.e. $\sigma$ is surjective. If $\sigma(x) = T$, we call $x$ a $\sigma$ $description$ of $T$.
\label{Prelim: Def: Binary Description}
\end{definition}

\noindent For a binary representation of F-transducers $\sigma$, we define $$|T|_{\sigma} = \min \{|x| : \sigma(x) = T\}$$ to be the \emph{size of {$T$} with respect to {$\sigma$}}. For all $k \in \N$, define $$\Fsize{k}_{\sigma} = \{ T \in \mathrm{F}:|T|_{\sigma} \leq k \}$$ to be the set of F-transducers with a $\sigma$ description of size $k$ or less. For all $k \in \N$ and $x \in \fbins$, the \emph{k-F complexity} of $x$ with respect to binary representation $\sigma$ is defined as $$D^{k}_{\sigma}(x) = \min \Big \{ |y| \, : \,  T \in \mathrm{F}^{\leq k}_{\sigma} \, \wedge \, T(y) = x \, \Big \}.$$ Here, $y$ is the shortest string that gives $x$ as an output when inputted into an F-transducer of size $k$ or less with respect to the binary representation $\sigma$. %$T$ can be thought of as the F-transducer that can decompress $y$ to reproduce $x$.
 
We fix an arbitrary binary representation of pebble transducers in this paper and write $\PBcomp{k}(x)$ to represent the $k$-PB complexity of string $x$. Note that this should not be confused with the set PB$_k$. We use the binary representation for FSTs presented in \cite{JordonM20}. It is chosen as it is used in the proof of Lemma \ref{FS: Lemma: GEQ Lemma} which requires an upper bound for the size of FSTs with the same transition and output functions, but with different start states. We omit the full details of the representation as if a sequence is FS-deep with respect to one representation, it is FS-deep with respect to all representations \cite{JordonMPDLZ}. Hence we write $\FScomp{k}(x)$ to represent the $k$-FS complexity of string $x$.

The upper and lower randomness density of a sequence $S$ for a family of transducers $F$ are given by $$\rho_F(S) = \lim_{k \to \infty}\liminf_{n \to \infty}\frac{\Fcomp{k}(S \upharpoonright n)}{n}, \textrm{ and } R_F(S) = \lim_{k \to \infty}\limsup_{n \to \infty}\frac{\Fcomp{k}(S \upharpoonright n)}{n} \textrm{ respectively.}$$ We say a sequence $S$ is $\Ft-$\emph{trivial} if $R_{\Ft}(S) = 0$ and $\Ft-$\emph{incompressible} if $\rho_{\Ft}(S) = 1$.

For every sequence $S$, the following result gives a relation between $k$-FS complexity of the prefixes of $S$ and the compression ratio of ILFSTs on prefixes of $S$.

\begin{theorem}[\cite{b.doty.moser.lossy.decompressors,DBLP:journals/iandc/SheinwaldLZ95}]
Let $S$ be a sequence. Then $$\rho_{FS}(S) = \inf_{C \in \ILFST}\limits\liminf_{n \to \infty}\frac{|C(S \upharpoonright n)|}{n}.$$
\label{compvs decomp}
\end{theorem}

Normal sequences are used in the proof of Theorem \ref{malicious}. It is well known that a sequence is normal iff it is incompressible by information lossless finite-state transducers (see \cite{DBLP:journals/tcs/BecherH13,b.lutz.finite-state.dimension,DBLP:journals/acta/SchnorrS72}). Combining this with Theorem \ref{compvs decomp} gives us the following corollary.

\begin{corollary}
Let $S$ be a sequence. Then $S$ is normal iff $\rho_{\textrm{FS}}(S) = 1.$
\label{normal fs}
\end{corollary}

The binary representation of FSTs we use satisfies the following lemma which is used, for example, in the proof of Lemma \ref{PB: Lemma: SGL No Work With PB}.
\begin{lemma}[\cite{JordonM20}]
There exists a binary representation of FSTs such that
%% PROOF IN APPENDIX
$$(\forall^{\infty} k \in \N)(\forall n \in \N)(\forall x,y,z \in \fbins) \FScomp{k}(xy^nz) \geq \FScomp{3k}(x) + n\FScomp{3k}(y) + \FScomp{3k}(z).$$
\label{FS: Lemma: GEQ Lemma}
\end{lemma}

\begin{remark}
Lemma \ref{FS: Lemma: GEQ Lemma} can be generalised such that for our fixed binary representation $\sigma$, we can break the input into any number of substrings to get a similar result. That is for any string $x = x_1\ldots x_n$,  $$(\forall^{\infty} k \in \N) \, \FScomp{k}(x_1 \ldots x_n) \geq \sum_{i=1}^n\FScomp{3k}(x_i).$$
\label{FS: Remark: any substring split}
\end{remark}

\section{Pebble-Depth}
\label{Pebble-Depth Section}

In this section we present our notion of \textit{pebble depth (PB-depth)}, show it satisfies the three basic fundamental properties of depth and identify a normal PB-deep sequence.

We define PB-depth by examining the difference in $k$-FS and $k'$-PB complexity on prefixes on sequences. This keeps to the spirit of Bennett's original notion of comparing Kolmogorov complexity against its restricted time-bounded version. FSTs can be viewed as a restricted subset of $0$-pebble transducers which can only move in one direction.

\begin{definition}
A sequence $S$ is \textit{pebble deep (PB-deep)} if $$(\exists \alpha > 0)( \forall k \in \N)( \exists k' \in \N)( \forall^{\infty} n \in \N) \FScomp{k}(S \upharpoonright n) - \PBcomp{k'}(S \upharpoonright n) \geq \alpha n.$$
\label{PB: Def: PB Depth}
\end{definition}

\begin{definition}
Let $S \in \infbins$. We say that PB-\textit{depth}$(S) \geq \alpha$ if $$(\forall k \in \N)(\exists k' \in \N)(\forall^{\infty} n \in \N) \, \FScomp{k}(S \upharpoonright n) - \PBcomp{k'}(S \upharpoonright n) \geq \alpha n.$$
    Otherwise we say PB-\textit{depth}$(S) < \alpha$.
    \label{PB: Def: PB depth}
\end{definition}

\subsection{Fundamental Properties of Pebble-Depth}

Before we show that pebble depth satisfies the fundamental properties of depth, we first show the existence of PB-incompressible sequences. The existence of PB-trivial sequences is evident in our later proofs. To demonstrate this, we first need the following result which relates H and K, the prefix-free and plain version of Kolmogorov complexity. It can be found as Corollary 2.4.2 in \cite{nies:book}.

\begin{lemma}
For all $x \in \fbins$ it holds that $$H(x) \leq K(x) + 2\log(K(x)) + O(1)\leq K(x) + 2\log(|x|) + O(1).$$
\label{PB: Lemma: H K Rel}
\end{lemma}

Martin-L{\"{o}}f randomness (ML-randomness)\cite{martinlof66} is used throughout algorithmic information theory as a way to define random sequences. While there are several equivalent characterisations of ML-randomness \cite{nies:book,downey:book}, we shall give the definition based on prefix-free Kolmogorov complexity.

\begin{definition}
$S \in \infbins$ is \emph{ML-random} if there is a constant $c \in \N$ such that for all $n \in \N$ it holds that $H(S \upharpoonright n) > n - c.$
\label{prelim: Def: 1 random}
\end{definition}

We next show ML-random sequences are PB-incompressible.

\begin{lemma}
If $S \in \infbins$ is ML-random, then $\rho_{\PB}(S) = 1.$
\label{PB: Lemma: pebble random}
\end{lemma}
\begin{proof}
Let $S \in \infbins$ be ML-random. Hence for all $n$ there exists some $c \in \N$ such that $H(S \upharpoonright n) > n-c.$

Fix $k \in \N$ and consider the prefix $S\upharpoonright n$ of $S$. Let $T \in \PBsize{k}$ and $y \in \fbins$ be such that $T(y) = S\upharpoonright n$ and $\PBcomp{k}(S \upharpoonright n) = |y|$. Let $M$ be the machine such that on inputs of the form $d(\sigma)01x$, where $\sigma$ is a description of a pebble transducer via our encoding and $x$ is in the domain of the pebble transducer that $\sigma$ describes, $M$ uses $d(\sigma)$ to retrieve the pebble transducer and then simulates the transducer on $x$. Otherwise, $M$ loops.

Hence we have that $$K(S \upharpoonright n) \leq 2|\sigma| + 2 + |y| + O(1) \leq 2k + 2 + \PBcomp{k}(S \upharpoonright n) + O(1).$$ By Lemma \ref{PB: Lemma: H K Rel} it follows that $$H(S \upharpoonright n) \leq 2k + 2 + \PBcomp{k}(S \upharpoonright n) + 2\log(n) + O(1) = \PBcomp{k}(S \upharpoonright n) + 2\log(n) + O(1).$$  Therefore $$\PBcomp{k}(S \upharpoonright n) > n - c - 2\log(n) - O(1) = n - 2\log(n) - O(1).$$ Note if no $y$ as above exists, this still holds as in such cases $\PBcomp{k}(S \upharpoonright n) = \infty$.

Therefore, for all $k$ we have that $$1 = \liminf\limits_{n \to \infty}\frac{n - 2\log(n) - O(1)}{n} \leq \liminf\limits_{n \to \infty}\frac{\PBcomp{k}(S \upharpoonright n)}{n} \leq 1.$$ As $k$ was arbitrary, it follows that $\rho_{\PB}(S) = 1.$

\end{proof}

The following demonstrates that sequences which are FST-trivial and PB-incompressible are not PB-deep. This is analogous to Bennett's fundamental properties of computable and ML-random sequences being shallow.

\begin{theorem}
Let $S \in \infbins$. If $R_{\mathrm{FS}}(S) = 0$ or $\rho_{\PB}(S) = 1$ then $S$ is not PB-deep.
\label{PB: Thm: Easy Hard}
\end{theorem}

\begin{proof}
Suppose that $R_{\textrm{FS}}(S) = 0.$ Let $\alpha> 0$. Let $k$ be such that for almost every $n$ \begin{equation}\FScomp{k}(S \upharpoonright n) \leq \alpha n. \label{PB: Thm: Easy Hard: EQ3}\end{equation} 

\noindent Then for all $k' \in \N$, for almost every $n$ we have that \begin{equation}
\FScomp{k}(S \upharpoonright n) - \PBcomp{k'}(S \upharpoonright n) \leq \FScomp{k}(S \upharpoonright n) < \alpha n. \label{PB: Thm: Easy Hard: EQ4}\end{equation}

\noindent As $\alpha$ was chosen arbitrarily, $S$ is not PB-deep.

Next suppose that $\rho_{\PB}(S) = 1$. Therefore for all $\alpha > 0$ and $k \in \N$, for almost every $n$ it holds that \begin{equation}
\PBcomp{k}(S \upharpoonright n) \geq (1 - \varepsilon). \label{PB: Thm: Easy Hard: EQ1}
\end{equation}

\noindent Therefore we have for $k'$ such that $I_{\mathrm{FS}} \in \FSTsize{k'}$,
\begin{equation}
\FScomp{k'}(S \upharpoonright n) - \PBcomp{k}(S \upharpoonright n) \leq n - (1 - \alpha)n = \alpha n. \label{PB: Thm: Easy Hard: EQ2} \end{equation} As $\alpha$ was chosen arbitrarily, $S$ is not PB-deep
 
\end{proof}

%Before we compare pebble-depth with other notions we first explore some properties of $k$-PB complexity and pebble-depth. The following straight forward lemma demonstrates that from a description of $x$, a pebble-transducer can be built which output $x^{|x|^n}$ using that same description of $x$.
%\begin{lemma}
%$(\forall k,n \in \N)(\exists k' \in \N)(\forall x \in \fbins) \PBcomp{k'}(x^{|x|^n}) \leq \PBcomp{k}(x).$
%\label{pt rep easy leq}
%\end{lemma}
%The proof is contained in the appendix.
%The following theorem states that if $p$ is a pebble description of $x$ and $q$ is a pebbble description of $y$, a a padded version of $p$ followed by a flag, followed by $q$ is a pebble description of the string $xy$.

%\begin{lemma}
%$(\forall \varepsilon > 0)(\forall k \in \N)(\exists k' \in \N)(\forall^{\infty} x \in \fbins)(\forall y \in \fbins)$ $$\PBcomp{k'}(xy) \leq (1 + \varepsilon)\PBcomp{k}(x) + \PBcomp{k}(y) + 2.$$\label{pt leq}
%\end{lemma}
%The proof is in the appendix.

Prior to showing the PB-depth satisfies a slow growth law, we require the following two lemmas which demonstrate relationships between $k$-FST and $k$-PB complexity of strings $x$ and $M(x)$, where $M$ is an ILFST. They are required for the proof of the slow growth law of PB-depth in Theorem \ref{PB: Thm: SGL}.

\begin{lemma}[\cite{DotyM07}]
Let $M$ be an $\mathrm{ILFST}$. 
     $$(\forall k \in \N)(\exists k' \in \N)(\forall x \in \fbins) \, \FScomp{k'}(M(x)) \leq \FScomp{k}(x).$$
\label{PB: Lemma: ILFST LEQs}
\end{lemma}

\begin{lemma}
Let $M \in \ILFST$. Then $$(\forall k \in \N)(\exists k' \in \N)(\forall x \in \N) \, \PBcomp{k'}(x) \leq \PBcomp{k}(M(x)).$$
\label{PB: Lemma: ILFST PB LEQs}
\end{lemma}

\begin{proof}
Let $M,\,k$ and $x$ be as stated in the lemma. By Theorem \ref{FS: Thm: inverse fst}, there exists an ILFST $M^{-1}$ and constant $b$ such that for all $z \in \fbins$, $z \upharpoonright |z| - b \sqsubseteq M^{-1}(M(x)) \sqsubseteq z$.

Note that both $M$ and $M^{-1}$ can be simulated by $0$-pebble transducers that print nothing on first reading $\dashv$, then read their input bit by bit moving right performing the same actions as $M$ and $M^{-1}$ respectively and where upon reading $\vdash$, they output nothing and enter their final state. For simplicity of notation we call these equivalent $0$-pebble transducers $M$ and $M^{-1}$ also.

Let $p$ be a $k$-PB minimal program for $M(x),$ i.e. $A(p) = M(x)$ for $A \in \mathrm{PB}^{\leq k}$, and $\PBcomp{k}(M(x)) = |p|$. We construct $A'$ and $p'$ for $x$. Let $y = M^{-1}(M(x)),$ i.e. there exists some $z \in \bin^{\leq b}$ such that $yz = x$. Let $A'$ be the PB which on input $p$ simulates $A(p)$ to get $M(x)$, and sticks the output into $M^{-1}$ and adds $z$ at the end of $M^{-1}$'s output, i.e. when it enters the final state of $M^{-1}$, regardless of what is under its reading head, it prints $z$ and enters its own final state. Note that by Theorem $\ref{PB: Thm: PB Composition}$, $A'$ will have the same number of pebbles as $A$. Thus $\PBcomp{|A'|}(x) \leq \PBcomp{k}(M(x))$. As $|A'|$ depends only on $k$, the size of $M$ and $|z|\leq b$, we set $k'$ to be the smallest integer that takes all the possibilities for $z$ into account. That is $\PBcomp{k'}(x) \leq \PBcomp{k}(M(x)).$  

\end{proof}

We can now demonstrate that PB-depth satisfies a slow growth law. In the following theorem we show that if a PB-deep sequence $S'$ is the output of some ILFS computable mapping, the original sequence $S$ used to compute $S'$ must also have been PB-deep. This is analogous to Bennett's slow growth law as it demonstrates that the fast process of computation by an ILFST cannot transform a non-deep sequence into a deep sequence.

\begin{theorem}[Slow Growth Law]
Let $S$ be a sequence. Let $f: \infbins \longrightarrow \infbins$ be $\mathrm{ILFS}$ computable, and let $S' = f(S)$. If $S'$ is PB-deep, then $S$ is PB-deep. 
\label{PB: Thm: SGL}
\end{theorem}

\begin{proof}
Let $S,\,S'$ and $f$ be as stated, and let $M$ be an ILFST computing $f$.

For all $n$ such that $M(S \upharpoonright m) = S'\upharpoonright n$ for some $m$, let $m_n$ denote the largest integer such that $M(S \upharpoonright m_n) = S' \upharpoonright n.$ Note then that for all $m$, there exists an $n$ such that $M(S \upharpoonright m_{n-1}) \sqsubset M(S \upharpoonright m) = M(S \upharpoonright m_n) = S' \upharpoonright n.$ As $M$ is IL, it cannot visit the same state twice without outputting at least one bit, so there exist a $\beta > 0$ such that for all $n$,
 $n \geq \beta m_n$.
 
 Fix $l \in \N$. Let $k$ be from Lemma \ref{PB: Lemma: ILFST LEQs} such that for all $x \in \fbins$ \begin{equation}
     \FScomp{k}(M(x)) \leq \FScomp{l}(x). \label{PB: Thm: SGL: EQ1}
 \end{equation}
 As $S'$ is PB-deep, there exists $k'$ and $\alpha > 0$ such that for almost every $n$ \begin{equation}
     \FScomp{k}(S' \upharpoonright n) - \PBcomp{k'}(S \upharpoonright n) \geq \alpha n. \label{PB: Thm: SGL: EQ2} 
 \end{equation}
 Similarly let $l'$ be from Lemma \ref{PB: Lemma: ILFST PB LEQs} such that for all $x$ \begin{equation}
     \PBcomp{l'}(x) \leq \PBcomp{k'}(M(x)). \label{PB: Thm: SGL: EQ3}
 \end{equation}
 
 Hence for almost every $m$ we have that
 \begin{align*}
     \FScomp{l}(S \upharpoonright m) - \PBcomp{l'}(S \upharpoonright m) & \geq \FScomp{k}(M(S \upharpoonright m) - \PBcomp{l'}(S \upharpoonright m) \tag{by \eqref{PB: Thm: SGL: EQ1}} \\
     & \geq \FScomp{k}(M(S \upharpoonright m) - \PBcomp{k'}(M(S \upharpoonright m)) \tag{by \eqref{PB: Thm: SGL: EQ3}} \\
     & = \FScomp{k}(M(S \upharpoonright m_n) - \PBcomp{k'}(M(S \upharpoonright m_n)) \tag{for some $n$} \\
     & = \FScomp{k}(S' \upharpoonright n) - \PBcomp{k'}(S' \upharpoonright n) \notag \\
     & \geq \alpha n \tag{by \eqref{PB: Thm: SGL: EQ2}} \\
     & \geq \alpha \beta m_n \geq \alpha \beta m. \notag
 \end{align*}
 Hence $S$ is PB-deep.
 
 \end{proof}
 
 Suppose that transformations via pebble transducers instead of via FSTs were used to define an alternative slow growth law. We show that the existence of a sequence $S$ such that $\rho_{\textrm{PB}}(S) = 1$ would break this alternative slow growth law. Before we do this, we need the following definition of a pebble-computable function.

\begin{definition}
A function $f: \infbins \rightarrow \infbins$ is said to be pebble computable if there is exists some $T \in \PB$ such that for all $S \in \infbins$, $\lim\limits_{n \to \infty}|T(S \upharpoonright n)| = \infty$ and for all $n \in \N$, $T(S \upharpoonright n) \sqsubseteq f(S)$.
\label{FS: Def: Pebble Computable}
\end{definition}

We now break the alternative slow growth law. The idea is to take a PB-incompressible sequence $S$, and to transform it into the sequence $S' = x_1x_2\ldots$ , where for each $i$, $x_i = S[0..i-1]$. On input $x_j$, a $1$-pebble transducer $T$ can print $x_1x_2x_3\cdots x_j$ by continuously moving its pebble one square to the right at each stage to keep track of the $x_i$'s it has printed. $T$ scans its head to the left end of the tape and then moves right printing what it sees up until it reaches the tape square containing the pebble upon which the pebble is moved one square right. Thus, prefixes of $S'$ have low pebble-complexity.

\begin{lemma}
Let $S \in \infbins$ be such that $\rho_{\textrm{PB}}(S) = 1$. There exists a pebble computable function $f$ such that for $S' = f(S)$, $S'$ is PB-deep while $S$ is not PB-deep.
\label{PB: Lemma: SGL No Work With PB}
\end{lemma}
\begin{proof}
Let $S \in \infbins$ be such that $\rho_{\PB}(S) = 1$. % for example, $S$ could be chosen to be an ML-random sequence. 
Recall the function $\pref$ on strings and note that $\pref$ can extended to be a pebble computable function on sequences. We will show that $S' = \pref(S)$ is PB-deep even though $S$ is not PB-deep by Theorem \ref{PB: Thm: Easy Hard}. 

Let $0<\varepsilon<1$. Note that as $\rho_{\PB}(S) = 1$, it follows that $\rho_{\textrm{FS}}(S) =1$ also. Hence, for all $k$ and almost every $n$ it holds that \begin{equation}
    \FScomp{k}(S \upharpoonright n) > n(1 - \frac{\varepsilon}{3}). \label{PB: Lemma: SGL No Work With PB: EQ1}
\end{equation}

Consider an arbitrary prefix  $S'\upharpoonright n$ of $S'$. Let $j$ be such that $$\frac{j(j+1)}{2} \leq n < \frac{(j+1)(j+2)}{2}.$$ $S' \upharpoonright n$ can be written in the form $x_1x_2\ldots x_j y$ where $x_i = S[0..i-1]$ and $y \sqsubset x_{j+1}$.

Suppose $j^*$ is such that Equation \eqref{PB: Lemma: SGL No Work With PB: EQ1} holds for all $j \geq j^*$. Then by Remark \ref{FS: Remark: any substring split} it follows that for almost all $k$ and large enough $n$, \begin{align}
    \FScomp{k}(S' \upharpoonright n) &\geq \sum_{i = 1}^j \FScomp{3k}(x_i) + \FScomp{3k}(y) \notag \\
    & > \sum_{i = j^*}^j \FScomp{3k}(x_i) \notag \\
    & > \sum_{i = j^*}^j|x_i|(1 - \frac{\varepsilon}{3}) \tag{by \eqref{PB: Lemma: SGL No Work With PB: EQ1}} \\ 
    & = (n - O(1) - |y|)(1 - \frac{\varepsilon}{3}) \notag \\
    & > n(1 - \frac{2\varepsilon}{3}) \label{PB: Lemma: SGL No Work With PB: EQ10}
\end{align}
as $|y| = O(\sqrt{n})$.

Similarly, let $T$ be the pebble transducer such that on inputs of the form $d(x)01z$, $T$ uses $d(x)$ and a pebble to print $\pref(x)$ and then uses $z$ to print $z$. Hence, $T(d(x_j)01y) = S'\upharpoonright n.$ Thus for $n$ large, \begin{equation}
    \PBcomp{|T|}(S'\upharpoonright n) \leq 2j + 2 + |y| < 3(j+1) = O(\sqrt{n}). \label{PB: Lemma: SGL No Work With PB: EQ2}
\end{equation}

Hence, for almost every $k$ and for $n$ large it holds that \begin{equation}
\FScomp{k}(S'\upharpoonright n) - \PBcomp{|T|}(S'\upharpoonright n) > n(1 - \frac{2\varepsilon}{3}) - \frac{n\varepsilon}{3} = n(1 - \varepsilon). \label{PB: Lemma: SGL No Work With PB: EQ3}
\end{equation}

As Equation \eqref{PB: Lemma: SGL No Work With PB: EQ3} in fact holds for every $k$ as $\FScomp{i}(x) \geq \FScomp{i+1}(x)$ for all strings $x$, $S'$ is in fact PB-deep. Thus the alternative slow growth law breaks.

For completeness, the construction of $T$ is confined to the appendix subsection \ref{cons1}.

\end{proof}

\section{Comparison with Other Depth Notions}

In this section we give the definition of finite-state, pushdown and LZ-depth. We continue by comparing PB-depth with these three notions.

In \cite{DotyM07}, the following definition is provided for an (infinitely often) FS-depth notion. An almost everywhere version is explored in \cite{JordonM20}

\begin{definition}
A sequence $S$ is finite-state deep (FS-deep) if 
$$(\exists \alpha > 0)(\forall k \in \N)(\exists k' \in \N)(\exists^{\infty} n \in \N) \, \FScomp{k}(S \upharpoonright n) - \FScomp{k'}(S \upharpoonright n) \geq \alpha n.$$ \label{fs deep definition}
\end{definition}

The following notions of PD-depth and LZ-depth are presented in \cite{JordonMPDLZ}.\footnote{We note that reference \cite{JordonMPDLZ} currently directs to a preprint version of an article.}

\begin{definition}
A sequence $S$ is pushdown-deep (PD-deep) if \begin{align*}(\exists \alpha > 0)(\forall C \in \ILUPDC)(\exists C' \in \ILPDC)(\forall^{\infty} n \in \N), 
 |C(S \upharpoonright n)| - |C'(S \upharpoonright n)| \geq \alpha n.\end{align*}
\end{definition}
%Note in the above definition, we deal exclusively with ILPDCs with $\fllog$-stack growth. $\fllog$ could be replaced with any family of order functions $F'$ to define a notion of PD$_{F'}$-depth as discussed in \cite{JordonMPDLZ}.

\begin{definition}
A sequence $S$ is \emph{Lempel-Ziv deep (LZ-deep)} if $$(\exists \alpha > 0)(\forall C \in \ILFST)(\forall^{\infty} n \in \N), |C(S \upharpoonright n)| - |LZ(S \upharpoonright n)| \geq \alpha n.$$
\end{definition}

For a sequence $S$, FS-depth$(S)$, PD-depth$(S)$ and LZ-depth$(S)$ are defined similarly to PB-depth$(S)$.

\subsection{Finite-State Depth}

In \cite{DotyM07}, Doty and Moser prove that no normal sequence is FS-deep in their notion. This is similarly true for Jordon and Moser's almost everywhere version \cite{JordonM20}. In this section, we demonstrate that a difference between FS-depth and PB-depth by showing that there are normal PB-deep sequences. To demonstrate this we require the following result by Lathrop and Strauss.

\begin{theorem}[\cite{DBLP:journals/iandc/LathropL99}]
There exists $S \in \infbins$ such that the sequence $S' = x_1x_2x_3\ldots$ where for each $i$, $x_i = S[0..i-1]$, is normal.
\label{malicious}
\end{theorem}

\begin{theorem}
There exists a normal sequence which is $PB$-deep.
\label{PB: Thm: Normal}
\end{theorem}
\begin{proof}

Let $S \in \infbins$ from Theorem \ref{malicious} such that the sequence $S' = x_1x_2x_3\ldots$ where for each $i$, $x_i = S[0..i-1]$, is normal. 

The proof that $S'$ is PB-deep follows a similar argument as the proof of Lemma \ref{PB: Lemma: SGL No Work With PB}, which also examined a sequence formed by concatenating prefixes of another sequence. The main difference is that Equation \eqref{PB: Lemma: SGL No Work With PB: EQ10} of the proof follows immediately from Corollary \ref{normal fs} as $S'$ is normal.

\end{proof}

\subsection{Lempel-Ziv Depth}

The following demonstrates the existence of a sequence $S$ with PB-depth of roughly $1/2$ and low LZ-depth. The sequence is that from Theorem $5$ of \cite{DBLP:journals/mst/MayordomoMP11}. This sequence is broken into blocks where each block is a concatenation of most strings of length $n$. Specifically blocks are composed of subblocks of the form $XFY$ where $X$ is a listing of a selection of strings of length $n$, $F$ is a flag not contained in any string of length $n$ listed, and $Y$ is a listing of strings of length $n$ such that $Y = X^{-1}$. A pebble transducer can perform well on this sequence as given $X$, the transducer can use its two-way tape property to print $Y$ also. LZ does not compress $S$ by much as it is almost a listing of every string in order of length. LZ compresses such sequences poorly.

\begin{theorem}
For each $0 < \beta < \frac{1}{2}$, there exists a sequence $S$ such that PB-depth$(S) \geq \frac{1}{2} - \beta$ and LZ-depth$(S) < \beta.$
\label{PB: Thm: PB vs LZ}
\end{theorem}

\begin{proof}

Let $0<\beta<\frac{1}{2}$, and let $k>2$ and $v$ be integers to be determined later. For any $n \in \N$, let $T_n$ denote the set of strings of length $n$ that do not contain the substring $1^j$ in $x$ for all $j \geq k$. As $T_n$ contains the set of strings whose every $k^{\text{th}}$ bit is $0$, it follows that $|T_n| \geq 2^{(\frac{k-1}{k})n}$. Note that for every $x \in T_n$, there exists $y \in T_{n-1}$ and $b \in \bin$ such that $x = yb$. Hence 
\begin{equation}
|T_n| < 2|T_{n-1}|. \label{PB: Thm: PB vs LZ: EQ1}
\end{equation}

Let $A_n = \{a_{1_n},\ldots,a_{u_n}\}$ be the set of palindromes in $T_n$. As fixing the first $\lceil \frac{n}{2} \rceil$ bits determines a palindrome, $|A_n| \leq 2^{\lceil \frac{n}{2} \rceil}$. The remaining strings in $T_n - A_n$ are split into $v+1$ pairs of sets $X_{n,i} = \{x_{n,i,1},\ldots, x_{n,i,t_n^i}\}$ and $Y_{n,i} = \{y_{n,i,1},\ldots, y_{n,i,t_n^i}\}$ where $t_n^i = \lfloor\frac{|T_n - A_n|}{2v} \rfloor$ if $i \neq v+1$ and $$t_n^{v+1} = \frac{1}{2}(|T_n-A_n| - 2\sum_{i = 1}^v|X_{n,i}|), $$ $(x_{n,i,j})^{-1} = y_{n,i,j}$ for every $1 \leq j \leq t_n^i$ and $1 \leq i \leq v+1$ both $x_{n,i,1}$ and $y_{n,i,t_n}$ start with $0$ (that is, $x_{n,i,t_n^i}$ ends with a $0$) excluding the case where both $X_{n,v+1}$ and $Y_{n,v+1}$ are the empty sets). Note that for convenience we write $X_i,Y_i$ for $X_{n,i},Y_{n,i}$ respectively.

$S$ is constructed in stages. Let $f(k) = 2k$ and $f(n+1) = f(n) + v + 2$. Note that $n < f(n) < n^2$ for large $n$. For $n \leq k-1$, $S_n$ is a concatenation of all strings of length $n$, i.e. $S_n = 0^n\cdot0^{n-1}1\cdots 1^{n-1}0 \cdot 1^n.$ For $n \geq k$,$$S_n = a_{1_n}\ldots a_{u_n}1^{f(n)}z_{n,1}z_{n,2}\ldots z_{n,v}z_{n,v+1}$$
    where $$z_{n,i} = x_{n,i,1}x_{n,i,2}\ldots x_{n,i,t^i_n-1}x_{n,i,t^i_n}1^{f(n) + i}y_{n,i,t^i_n}y_{n,i,t^i_n-1}\ldots y_{n,i,2}y_{n,i,1},$$ with the possibility that $z_{n,v+1} = 1^{f(n) + v + 1}$ only.
That is, $S_n$ is a concatenation of all strings in $A_n$ followed by a flag of $f(n)$ ones, followed by a concatenation of all strings in the $X_i$ zones and $Y_i$ zones separated by flags of increasing length such that each $Y_i$ zone is the $X_i$ zone written in reverse. Let $$S = S_1S_2\ldots S_{k-1}1^k1^{k+1}\ldots 1^{2k-1}S_kS_{k+1}\ldots$$ i.e. the concatenation of all $S_j$ zones with some extra flags between $S_{k-1}$ and $S_k$.

We first examine the lower randomness density of $S$ for pebble transducers.
\begin{claim}
$$\lim_{k \to \infty}\limsup_{n \to \infty}\frac{\PBcomp{k}(S \upharpoonright n)}{n} \leq \frac{1}{2}.$$
\label{PB: Thm: PB vs LZ: Claim1}
\end{claim}

\noindent We prove this claim by building the $1$-pebble transducer $T$ that acts as follows: $T$ begins moving right and printing its input until it sees the first $0$ after a flag of $2k$ ones. Upon seeing this $0$, if the succeeding bit is a $1$, $T$ stays in the print zone. $T$ moves right and prints what is on its tape until it sees a flag of $2k$ ones followed by a $0$ again. If $T$ sees a $0$ after $1^{2k}0$, $T$ enters a print-and-reverse zone. $T$ drops its pebble on the succeeding square. $T$ moves its head right printing what it sees until it sees $1^{2k}0$ (without printing the last $0$), then scans left past the flag of $1$s. Once the flag of $1$s ends, $T$ prints what it sees (i.e. printing the reverse of what it just printed) until it reaches the square with the pebble, printing what is on it. $T$ then moves right until it sees $1^{2k}0$ again and checks the next bit to see if it is in a print or print-and-reverse zone.

Let $y = S_1 \ldots S_{k-1}1^k \ldots 1^{2k-1}$. Then $T(y0) = y.$ Note that $|y| +1 < 2^{2k}.$ For $n \geq k$ and $1 \leq i \leq v$, let $$\pi_n = 1a_{1_n}\ldots a_{u_n}1^{f(n)}0 \, \textrm{ and } \, \sigma_{n,i} = 0x_{n,i,1}\cdots x_{n,i,t_n^i}1^{f(n) + i}0.$$ If $i = v+1$ we let $$\sigma_{n,v+1} = \begin{cases}
    0x_{n,v+1,1}\cdots x_{n,v+1,t_n^{v+1}}1^{f(n) + v+1}0, & \textrm{if} \, |X_{n,v+1}| \neq 0,\\
    1 1^{f(n) + v + 1}0, & \textrm{ otherwise}.
\end{cases}$$ 
Lastly we set $$\tau_n = \pi_n\sigma_{n,1}\sigma_{n,2}\ldots \sigma_{n,v+1}.$$ Note that \begin{align}
    |\tau_n| &\leq |A_n|n + (v+2)(f(n) + v+1) + 2(v + 2) + \frac{n}{2}|T_n - A_n| \notag\\
    & = |A_n|n + (v+2)(f(n) + v+3) + \frac{n}{2}|T_n - A_n|. \label{PB: Thm: PB vs LZ: EQ2}
    \end{align} 

\noindent Then as $T(y0\tau_k \ldots \tau_{n-1}) = S_1 \ldots S_{k-1}1^k\ldots 1^{2k-1} S_k\ldots S_{n-1},$ it follows that
\begin{align}
    \PBcomp{|T|}( S_1 \ldots S_{k-1}1^k\ldots 1^{2k-1} S_k \ldots S_{n-1}) \leq |S_1 \ldots S_{k-1}1^k\ldots 1^{2k-1}| + 1 \notag\\
    + \sum_{j=k}^{n-1} [|A_j|j + (v+2)(f(j) + v + 3)  +\frac{j}{2}|T_j - A_j|]. \label{PB: Thm: PB vs LZ: EQ3}
\end{align}

Let $w_p$ be the string such that $T(w_p) = S\upharpoonright p.$ Note that the ratio $\frac{|w_p|}{|p|}$ is maximal if the suffix of $S \upharpoonright p$ is a full concatenation of a $Y_{n,i}$ zone without the final bit. That is, $S \upharpoonright p$ ends with a suffix of the form $$y_{n,i,t_n^i}\ldots y_{n,i,2}y_{n,i,1}[0..n-2].$$ This is because $T$ cannot make use of its two-way capability to print the reverse of the $X$ zone since it does not know where to stop. In particular, the ratio is maximal on the zone $i=1$ as it immediately follows palindrome portion of $S_n$ where $T$ acts as the identity transducer to output it.

Let $0 \leq I < v$. We do not examine the case where $I = v$ as in this case, $T$ requires the fewest amount of bits to output the $v+1^{\thh}$ zone. We examine the ratio $\frac{|w_p|}{|S \upharpoonright p|}$ inside zone $S_n$ on the second last symbol of the $Y_{I+1}$ zone. Note that $T$ outputs $S \upharpoonright p$ on input
\begin{align*}
    &y0\tau_k \ldots \tau_{n-1} \pi_n\sigma_{n,1}\ldots \sigma_{n,I}1z
\end{align*}
where $$z = x_{n,I+1,1}\ldots
    x_{n,I+1,t_n}1^{f(n)+I+1}y_{n,I+1,t_n}\ldots y_{n,I+1,2}y_{n,I+1,1}[0..n-2].$$
Thus
\begin{align}
    |w_p| &\leq 2^{2k} + \sum_{j=k}^{n-1} [|A_j|j + (v+2)(f(j) + v + 3)
    +\frac{j}{2}|T_j - A_j|] \notag\\
    & + |A_n|n + (v+2)(f(n) + v + 3) + I(\frac{n|T_n - A_n|}{2v}) + \frac{n|T_n - A_n|}{v} \label{PB: Thm: PB vs LZ: EQ4}
    %& \leq 2^{cn} + \sum_{j=k}^{n-1}\frac{j}{2}|T_j| + \frac{n|T_n|}{v}(\frac{I}{2} + 1)\\
    %& = 2^{cn} + \sum_{j=k}^{n-1}\frac{j}{2}|T_j| + \frac{n|T_n|}{2v}(I + 2),
\end{align}    

 %%
 % It's I+1 here and not I+1/2 as in poly paper as here we end in a Y zone, not an X
 %
 % DO THIS TOMORROW
 %%
 
 Note first that \begin{equation}
     \sum_{j=k}^{n}|A_j|j \leq n^2|A_n| \leq n^2\cdot2^{\lceil\frac{n}{2}\rceil} \leq n^2 \cdot2^{\frac{n+1}{2}} \label{PB: Thm: PB vs LZ: EQ5}
 \end{equation}for $n$ large. Similarly the summation of the $(f(j)+v+2)$ contributes at most a polynomial number of bits in $n$. Along with the $2^{2k}$ term being a constant term this gives us for all $\varepsilon > 0$, for $n$ large 
 \begin{align}
     |w_p| &\leq 2^{n(\frac{1}{2} + \varepsilon)} + \sum_{j=k}^{n-1} \frac{j}{2}|T_j| + \frac{n|T_n|}{v}(\frac{I}{2} + 1) \notag \\
     & = 2^{n(\frac{1}{2} + \varepsilon)} + \sum_{j=k}^{n-1} \frac{j}{2}|T_j| + \frac{n|T_n|}{2v}(I + 2). \label{PB: Thm: PB vs LZ: EQ6}
 \end{align}
 The number of bits in such a prefix of $S$ is
  
\begin{align}
      |S \upharpoonright p| & \geq \sum_{j = k}^{n-1}j|T_j| + n|A_n| +  2n\Big \lfloor \frac{|T_n - A_n|}{2v}\Big\rfloor(I  + 1) \notag \\
      & \geq \sum_{j = k}^{n-1} + n|A_n| + 2n(\frac{|T_n-A_n|}{2v}-1)(I+1) \notag \\
      & = \sum_{j = k}^{n-1} + n|A_n| + n(\frac{|T_n|-|A_n|}{v}-2)(I+1) \notag\\
      & = \sum_{j = k}^{n-1}j |T_j| + n|A_n|(1 - \frac{(I+1)}{v}) + n(I+1)(\frac{|T_n|}{v} -2) \notag \\
      & \geq \sum_{j = k}^{n-1}j |T_j| + \frac{n}{v}|T_n|(I) \label{PB: Thm: PB vs LZ: EQ7}
  \end{align}
   as $I+1 \leq v$.
  %\begin{align}
   %   |S \upharpoonright p| & \geq \sum_{j = k}^{n-1}j|T_j| + n|A_n| + \big\lfloor \frac{n}{v}\big\rfloor |T_n - A_n|(I  + 1) \notag \\
    %  & \geq \sum_{j = k}^{n-1} + n|A_n| + (\frac{n}{v}-1)|T_n - A_n|(I+1) \notag \\
     % & \geq \sum_{j = k}^{n-1}j |T_j| + n|A_n|(1 - (\frac{1}{v} - \frac{1}{n})(I+1)) + \frac{n}{v}|T_n|I \notag \\
   %   & \geq \sum_{j = k}^{n-1}j |T_j| + n|A_n|(1 - \frac{I+1}{v}) + \frac{n}{v}|T_n|I \notag \\
 %     & \geq \sum_{j = k}^{n-1}j |T_j| + \frac{n}{v}|T_n|(I) \label{PB: Thm: PB vs LZ: EQ7}
 % \end{align}
  % as $I+1 \leq v$.
  
 %The number of bits in such a prefix of $S$ is
 %\begin{align}
  %   |S \upharpoonright p| & \geq \sum_{j=k}^{n-1}j|T_j| + n|A_n| + \frac{n}{v}|T_n - A_n|(I + 1) \notag\\
   %  & = \sum_{j=k}^{n-1}j|T_j| + \frac{n}{v}|T_n|(I + 1) + n|A_n|(1 - \frac{(I+1)}{v}) \notag \\
%     & \geq \sum_{j=k}^{n-1}j|T_j| + \frac{n}{v}|T_n|(I + 1) \label{PB: Thm: PB vs LZ: EQ7}
% \end{align}
% as $I+1 \leq v$.
  %%
 % It's I+1 here and not I+1/2 as in poly paper as here we end in a Y zone, not an X
 %
 %%
 
 Hence,
 \begin{align}
     \limsup\limits_{n \rightarrow \infty}\frac{|w_n|}{|S \upharpoonright n|} & \leq \limsup\limits_{n \rightarrow \infty}\frac{2^{n(\frac{1}{2} + \varepsilon)} + \sum_{j=k}^{n-1} \frac{j}{2}|T_j| + \frac{n|T_n|}{2v}(I + 2)}{\sum_{j = k}^{n-1}j|T_j| + \frac{n|T_n|}{v}(I)} \tag{by \eqref{PB: Thm: PB vs LZ: EQ6} and \eqref{PB: Thm: PB vs LZ: EQ7}}\\
     & = \limsup\limits_{n \rightarrow \infty}\bigg[\frac{2^{n(\frac{1}{2} + \varepsilon)} + 2\cdot\frac{n|T_n|}{2v}}{\sum_{j = k}^{n-1}j|T_j| + \frac{n|T_n|}{v}(I)} \notag \\ 
    & + \frac{1}{2}\cdot\frac{\sum_{j=k}^{n-1} j|T_j| + \frac{n|T_n|}{v}(I)}{\sum_{j = k}^{n-1}j|T_j| + \frac{n|T_n|}{v}(I)} \bigg] \notag \\
    % &+ \frac{\frac{n|T_n|}{2v}}{\sum_{j = k}^{n-1}j|T_j| + \frac{n|T_n|}{v}(I + 1)}\bigg]   \notag \\ 
     & = \limsup\limits_{n \rightarrow \infty}\bigg[\frac{2^{n(\frac{1}{2} + \varepsilon)} + \frac{n|T_n|}{v}}{\sum_{j = k}^{n-1}j|T_j| + \frac{n|T_n|}{v}(I)} + \frac{1}{2} \bigg]. \label{PB: Thm: PB vs LZ: EQ8}
      %+ \frac{\frac{n|T_n|}{2v}}{\sum_{j = k}^{n-1}j|T_j| + \frac{n|T_n|}{v}(I + 1)}\bigg]. \label{pd comp 1}
 \end{align}
 By \eqref{PB: Thm: PB vs LZ: EQ1}, as $\sum_{j=k}^{n-1}j|T_j| \geq (n-1)|T_{n-1}| \geq \frac{(n-1)}{2}|T_n|$, we have
 \begin{align}
     \sum_{j=k}^{n-1}j|T_j| + \frac{n}{v}|T_n|(I)&\geq \frac{n-1}{2}|T_n| +\frac{n}{v}|T_n|(I) \notag\\
     & = \frac{n|T_n|}{2v} (2I + v - \frac{v}{n}). \label{PB: Thm: PB vs LZ: EQ9}
 \end{align}
 Thus, when $\varepsilon$ is chosen to be such that $0 < \varepsilon < \frac{1}{2}- \frac{1}{k}$ we have that
\begin{align*}
    \limsup\limits_{n \rightarrow \infty} \frac{2^{n(\frac{1}{2} + \varepsilon)}}{\sum_{j = k}^{n-1}j|T_j| + \frac{n|T_n|}{v}(I)} & \leq \limsup\limits_{n \rightarrow \infty}\frac{2^{n(\frac{1}{2} + \varepsilon)}}{\frac{(n-1)}{2}|T_n|} 
     \leq \limsup\limits_{n \rightarrow \infty}\frac{2^{n(\frac{1}{2} + \varepsilon)}}{|T_n|} \\
    & \leq \limsup\limits_{n \rightarrow \infty}\frac{2^{n(\frac{1}{2} + \varepsilon)}}{2^{\frac{(k-1)n}{k}}} = 0 \label{PB: Thm: PB vs LZ: EQ10}
\end{align*}
as $k > 2$. Similarly by \eqref{PB: Thm: PB vs LZ: EQ9} we have
\begin{equation}
\frac{\frac{n|T_n|}{v}}{\sum_{j = k}^{n-1}j|T_j| + \frac{n|T_n|}{v}(I + 1)} \leq \frac{\frac{n|T_n|}{v}}{\frac{n|T_n|}{2v}(2I + v - \frac{v}{n})} 
 \leq \frac{2}{v(1 - \frac{1}{n})} \label{PB: Thm: PB vs LZ: EQ11}
\end{equation}
which can be made arbitrarily small by choosing $v$ appropriately large.

Therefore $$\limsup\limits_{n \rightarrow \infty}\frac{|w_n|}{|S \upharpoonright n|} \leq \frac{1}{2}.$$ This establishes Claim \ref{PB: Thm: PB vs LZ: Claim1}, i.e. for all $0<\beta'< 1/2 - 3/k$, we can choose $v,\,I$ and $k$ appropriately such that \begin{align}
    \PBcomp{|T|}(S \upharpoonright n) \leq (\frac{1}{2} + \frac{\beta'}{2})n. \label{PB: Thm: PB vs LZ: EQ12}
\end{align}

Next we examine how well any ILFST can compress prefixes of $S$. We use Theorem \ref{compvs decomp} to relate the compression performance back to $k$-finite state complexity.

\begin{claim}
$$\rho_{\textrm{FS}}(S) \geq \frac{k-3}{k}.$$
\label{PB: Thm: PB vs LZ: Claim2}
\end{claim}

Let $C \in \text{ILFST}.$ We assume every state in $C$ is reachable from its start state. Let $n\geq k$ and suppose $C$ is reading $S_n$. We examine the proportion of strings in $T_n$ that give a large contribution to the output. The argument is similar to that found in \cite{DBLP:journals/tcs/BecherH13}. 

We write $C(p,x) = (q,v)$ to represent that when $C$ is in state $p$ and reads input $x$, $C$ outputs $v$ and finishes in state $q$. For each $x \in T_n$, let
$$h_x = \min \{ |v| : \exists p,q \in Q, C(p,x) = (q,v)\}$$ be the minimum possible addition of the output that could result from $C$ reading $x$. Let $$B_n = \{ x \in T_n : h_x \geq \frac{(k-2)n}{k} \} $$ be the `incompressible' strings that give a large contribution to the output.

We write $C(p,x) = (q,v)$ to represent that when $C$ is in state $p$ and reads input $x$, $C$ outputs $v$ and finishes in state $q$, i.e. $C(p,x) = (\delta_C(p,x), \nu_C(p,x)) = (q,v)$. For each $x \in T_n$, let
$$h_x = \min \{ |v| : \exists p,q \in Q, C(p,x) = (q,v)\}$$ be the minimum possible addition of the output that could result from $C$ reading $x$. Let $$B_n = \{ x \in T_n : h_x \geq \frac{(k-2)n}{k} \} $$ be the `incompressible' strings that give a large contribution to the output.

Either $B_n = T_n$ meaning that $|B_n| = |T_n|$ or $B_n$ is a strict subset of $T_n$. If the latter case is true,  consider $x' \in T_n - B_n$. Then there is a computation of $x'$ that results in $C$ outputting at most $\frac{(k-2)n}{k}$ bits. As $C$ is lossless, $x'$ can be associated uniquely to a start state $p_{x'}$, end state $q_{x'}$ and output $v_{x'}$ where $|v_{x'}| < \frac{(k-2)n}{k}$ such that $C(p_{x'},x') = (q_{x'},v_{x'}).$ That is, we can build an injective map $g: T_n - B_n \to Q \times \bin^{< \frac{(k-2)n}{k}} \times Q$ where $g(x') = (p_{x'},v_{x'},q_{x'})$. As this map $g$ is injective, we can bound $|T_n - B_n|$ from above by
\begin{equation}
   | T_n - B_n| < |Q|^2 \cdot 2^{\frac{(k-2)n}{k}}. \label{PB: Thm: PB vs LZ: EQ13}
\end{equation}

Let $0 < \delta < \frac{1}{12(k-2)}$.  As $|T_n| \geq 2^{\frac{(k-1)n}{k}}$, by \eqref{PB: Thm: PB vs LZ: EQ13} we have that for $n$ large
\begin{align}
    |B_n| &= |T_n| - |T_n - B_n| \notag\\
    & > |T_n| - |Q|^2  \cdot 2^{\frac{(k-2)n}{k}} \notag \\
    & > |T_n|(1 - \delta). \label{PB: Thm: PB vs LZ: EQ14} 
\end{align}

Similarly, as the flags only compose $O(n^2)$ bits in each $S_n$ zone for $n \geq k$, we have for $n$ large that  
\begin{align}
|T_n|n > |S_n|(1 - \delta). \label{PB: Thm: PB vs LZ: EQ15}
\end{align}

Then for $n$ large (say for all $n\geq i$ such that \eqref{PB: Thm: PB vs LZ: EQ14} and \eqref{PB: Thm: PB vs LZ: EQ15} hold),
\begin{align}
    |C(S_1 \ldots S_i \ldots S_n)|
    & >  \frac{k-2}{k}\sum_{j=i}^m j |B_j| \notag\\
    & > \frac{k-2}{k}(1 - \delta)\sum_{j=i}^n j|T_j|  \tag{by \eqref{PB: Thm: PB vs LZ: EQ14}}\\
    & > \frac{k-2}{k}(1 - 2\delta)\sum_{j=i}^n |S_j| \tag{by \eqref{PB: Thm: PB vs LZ: EQ15}}\\
   % & > \frac{k-2}{k}(1 - 2\delta)|S_i \ldots S_m| \\
    & = \frac{k-2}{k}(1 - 2\delta)(|S_1 \ldots S_n| - |S_1 \ldots S_{i-1}|) \notag \\
    & > \frac{k-2}{k}(1 - 3\delta)|S_1 \ldots S_n|. \label{PB: Thm: PB vs LZ: EQ16}
\end{align}

The compression ratio of $S$ on $C$ is least on prefixes of the form $S_1 \ldots S_n x_{n+1}$, where potentially $x_{n+1}$ is a concatenation of all the strings in $T_{n+1}-B_{n+1}$, i.e. the compressible strings of $T_{n+1}$. Let $x_{n+1}$ be a such a potential prefix of $S_{n+1}$. Then if $F_{n+1} = \sum_{i=0}^{v+1} (f(n+1) + i)$ is the length of the flags in $S_{n+1}$, we can bound the length of $|x_{n+1}|$ as follows: \begin{align}
    |x_{n+1}| &< |T_{n+1}-B_{n+1}|(n+1) + F_{n+1} \notag\\
    & < (|T_{n+1}| - |B_{n+1}|)(n+1) + cn^2 \tag{for some  $c \in \N$} \\
    &< \delta|T_{n+1}|(n+1) + \delta|T_{n}|(n+1) \tag{by \eqref{PB: Thm: PB vs LZ: EQ14}} \\
    & < 2\delta|T_n|(n+1) + \delta|T_{n}|(n+1)  \tag{by \eqref{PB: Thm: PB vs LZ: EQ1}}\\
    & = 3\delta|T_n|(n+1)
     < 3\delta|S_n| + 3\delta|T_n| < 3\delta|S_n| + 6\delta|T_{n-1}| \tag{by \eqref{PB: Thm: PB vs LZ: EQ1}} \\
    & < 6\delta|S_1 \ldots S_n| \label{PB: Thm: PB vs LZ: EQ17}
\end{align}
for $n$ large.

So for $n$ large,
\begin{align*}
    |C(S_1 \ldots S_nx_{n+1})| &> (\frac{k-2}{k})(1 - 3\delta)(|S_1 \ldots S_nx_{n+1}| - |x_{n+1}|)  \tag{by \eqref{PB: Thm: PB vs LZ: EQ16}}\\
    & > \frac{k-2}{k}(1 - 6\delta)(|S_1 \ldots S_nx_{n+1}| - 6\delta|S_1 \ldots S_n|) \tag{by \eqref{PB: Thm: PB vs LZ: EQ17}} \\
   % &= \frac{k-2}{k}(1 - 6\delta)(|S_1 \ldots S_n|(1 - 6\delta) + |x_{n+1}|) \\ 
    &> \frac{k-2}{k}(1 - 6\delta)(|S_1 \ldots S_n|(1 - 6\delta) + |x_{n+1}|(1 - 6\delta)) \\ 
    &  > \frac{k-2}{k}(1 - 6\delta)^2|S_1 \ldots S_nx_{n+1}| > \frac{k-2}{k}(1 - 12\delta)|S_1 \ldots S_nx_{n+1}|\\
    & > \frac{k-3}{k}|S_1 \ldots S_nx_{n+1}| \label{PB: Thm: PB vs LZ: EQ18}
\end{align*}
by choice of $\delta$. 

Hence $\rho_{\textrm{FS}}(S) \geq \frac{k-3}{k}$ establishing Claim \ref{PB: Thm: PB vs LZ: Claim2}.Therefore for all $m$, for almost every prefix of $S$ we have that $$\FScomp{m}(S \upharpoonright n) \geq (\frac{k-3}{k} - \frac{\beta'}{2}).$$

Hence we have that for all $m$ and for almost every prefix of $S$ that \begin{equation}
    \FScomp{m}(S \upharpoonright n) - \PBcomp{|T|}(S \upharpoonright n) \geq (\frac{k-3}{k} -\frac{\beta'}{2})n  - (\frac{1}{2} + \frac{\beta'}{2})n = (\frac{1}{2} - \frac{3}{k} - \beta')n. \label{PB: Thm: PB vs LZ: EQ19} 
\end{equation}

Then, choosing $k,v$ and $\beta'$ such that $\beta =\frac{3}{k} + \beta' < \frac{1}{2}$ gives us the desired result that PB-depth$(S) \geq \frac{1}{2} - \beta.$

Next we examine LZ-depth. Recall $\rho_{LZ}(S) \geq 1 - \varepsilon$. Thus for $c$ such that $\varepsilon + c < \beta$ (recall $\varepsilon < \beta)$, for almost every $n$ it holds that \begin{equation}|LZ(S \upharpoonright n)| > (1 - \varepsilon-c)n. \label{LZ: Thm: PD not LZ: EQ10}\end{equation} Hence as $I_{\mathrm{FS}} \in \ILFST$, we have that for almost every $n$ \begin{equation}|I_{\mathrm{FS}}(S \upharpoonright n)| - |LZ(S \upharpoonright n)| < n - (1 - \varepsilon-c)n = (\varepsilon + c)n < \beta n. \label{LZ: Thm: PD not LZ: EQ11}\end{equation} Hence we have that LZ-depth$(S) < \beta$.

In conclusion, for all $0 < \beta < \frac{1}{2}$, choosing $\varepsilon$ such that $\varepsilon<\beta$ and $k$ such that $\frac{4}{k} < \beta$, a sequence $S$ can be built which satisfies the requirements of the theorem.

For completeness, the construction of $T$ is confined to the appendix subsection \ref{cons2}.

\end{proof}

\begin{corollary}
There exists a non-normal PB-deep sequence.
\label{PB: Cor: Non normal deep}
\end{corollary}

\begin{proof}
This follows from Theorem \ref{PB: Thm: PB vs LZ} since the string $01^{2k}0$ only occurs as a substring of the constructed $S$ a finite number of times. This is clear as the only places $01^{2k}$ can occur is if $0$ is the last bit of $S_{k-1}$ or where the $1^{2k}$ is a prefix to a flag in some zone $S_n$. However, as the flags increase in length, $01^{2k}$ will eventually always be followed by another $1$.

\end{proof}

\subsection{Pushdown Depth - Preliminary Result}

In this section we do not present an example of a sequence which is PB-deep but not PD-deep. More work is to be done to find such sequences, if they exist. Instead we present a preliminary result which states that for all $0< \beta < 1/2$, one can construct a sequence $S$ such that PB-depth$(S) \geq 1 - \beta$ while PD-depth$(S) \geq 1/2 - \beta$. Hence, the sequence is deep in both notions, and it is possible that their depth levels are in fact equal.

The sequence is composed of strings of the form $R^{|R|}F(R^{-1})^{|R|}$ where $F$ is a flag and $R$ is a string not containing $F$ with large plain Kolmogorov complexity relative to its length. Note that $R^{|R|}$ is a string of length $|R|^2$. From a single description of $R$, a $1$-pebble transducer can use a single pebble to print $R^{|R|}$. A large ILPDC with no restriction on its stack can be built to push $R^{|R|}$ onto its stack, and then when it sees the flag $F$, use its stack to compress $(R^{-1})^{|R|}$. These $R$ are built such that an $\ILUPDC$ is unable to use its stack to compress $R$, resulting in minimal compression.

\begin{remark}
For all $0< \beta < 1/2$, there exists a sequence $S$ such that  PB-depth$(S) \geq 1 - \beta$ and PD-depth$(S) \geq 1/2 - \beta$.
\label{PB: rmk: pb not pdfs}
\end{remark}

\begin{proof}

Let $0 < \beta < 1/2$ and let $k> 8$ be such that $\beta \geq 8/k$. For each $n$, let $t_n = k^{\lceil\frac{\log n}{\log k}\rceil}$. Note that for all $n$, \begin{align}
    n \leq t_n \leq kn. \label{PB: rmk: pb not pdfs2: EQ1} 
\end{align}
Consider the set $T_j$ which contains all strings of length $j$ that do not contain $1^{k}$ as a substring. As $T_j$ contains strings of the form $x_10x_20x_30\cdots$ where each $x_t$ is a string of length $k-1$, we have that $|T_j| \geq 2^{j(1 - \frac{1}{k})}$.
For each $j$, let $R_j \in \bin^{kt_j}$ have maximal plain Kolmogorov complexity in the sense that 
\begin{equation}
K(R_j) \geq |R_j|(1 - \frac{1}{k}). \label{PB: rmk: pb not pdfs2: EQ2}
\end{equation} Such an $R_j$ exists as $|T_{|R_j|}| > 2^{|R_j|(1 - \frac{1}{k})}-1$. Note that $kj \leq |R_j| \leq k^2j$. We construct $S$ in stages $S = S_1S_2\ldots$ where for each $j$, $$S_j = R_j^{|R_j|}1^k(R_j^{-1})^{|R_j|}.$$

\begin{claim}
PD-depth$(S) \geq \frac{1}{2} - \beta.$
\end{claim}

First we examine how well any ILUPDC compresses occurrences of $R_j$ zones in $S$. Let $C \in \ILUPDC$. Consider the tuple  $$(\widehat{C},q_s,q_e,z,\nu_C(q_s,R_j,z))$$
where $\widehat{C}$ is an encoding of $C$, $q_s$ is the state that $C$ begins reading $R_j$ in, $q_e$ is the state $C$ ends up in after reading $R_j$, $z$ is the stack contents of $C$ as it begins reading $R_j$ in $q_s$ (i.e. $z = 0^pz_0$ for some $p$), and the output $\nu_C(q_s,R_j,z)$ of $C$ on $R_j$. By Remark \ref{PD: Rmk: height discussion}, $C$'s stack is only important if $|z| < (c+1)|R_j|$, as if $|z|$ is larger, $C$ will output the same irregardless of $|z|$'s true value.  Hence, setting \begin{equation}
    z' = \begin{cases}
         |z| & \text{ if } |z| < (c+1)|R_j| \\
         (c+1)|R_j| & \text{ if } |z| \geq (c+1)|R_j|,
    \end{cases}
\end{equation}
as $C$ is lossless, having knowledge of the tuple $(\widehat{C},q_s,q_e,z',\nu_C(q_s,R_j,z))$ means we can recover $R_j$.

Using the fact that tuples of the form $(x_1,x_2,\ldots,x_n)$ can be encoded by the string \begin{equation}1^{\lceil \log n_1 \rceil}0n_1x_11^{\lceil \log n_2 \rceil}0n_2x_2 \ldots1^{\lceil \log n_{n-1} \rceil}0n_{n-1}x_{n-1}x_n,\label{Tuple Encoding}\end{equation} where $n_i = |x_i|$ in binary, and noting that $z'$ contributes roughly $O(\log |R_j|)$ bits to the encoding, we have we have by Equation \eqref{PB: rmk: pb not pdfs2: EQ2} that \begin{equation}
    |R_j|(1 - \frac{1}{k}) \leq K(R_j) \leq |\nu_C(q_s,R_j,z)| + O(\log|R_j|) + O(|\widehat{C}|) + O(1). \label{PB: rmk: pb not pdfs2: EQ3}
\end{equation}
Therefore, for $j$ large we have \begin{align}
    |\nu_C(q_s,R_j,z)| \geq |R_j|(1 - \frac{1}{k}) - O(\log|R_j|) > |R_j|(1 - \frac{2}{k}) \label{PB: rmk: pb not pdfs2: EQ4}
\end{align}
This is similarly true for $R_j^{-1}$ zones also as $K(R_j) \leq K(R_j^{-1}) + O(1)$. Hence for $j$ large we see that $C$ outputs at least
\begin{align}
    |C(\overline{S_j})| - |C(\overline{S_{j-1}})| &\geq 2|R_j|^2(1 - \frac{2}{k}) \notag\\
    & = (|S_j| - k)(1 - \frac{2}{k}) \notag \\
    & \geq |S_j|(1 - \frac{3}{k}) \label{PB: rmk: pb not pdfs2: EQ5}
\end{align}
bits when reading $S_j$.

Next we examine how well $C$ compresses $S$ on arbitrary prefixes. Consider the prefix $S \upharpoonright n$ and let $j$ be such that $\overline{S_j}$ is a prefix of $S \upharpoonright n$ but $\overline{S_{j+1}}$ is not. Thus $S\upharpoonright n = \overline{S_j}\cdot y$ for some $y \sqsubset S_{j+1}$. Suppose Equation \eqref{PB: rmk: pb not pdfs2: EQ5} holds for all $i\geq j'$. Hence we have that
\begin{align}
    |C(S \upharpoonright n)| & \geq |C(\overline{S_j})| \notag \geq |C(\overline{S_j})| - |C(\overline{S_{j'-1}})| \notag \\
    & \geq |S_{j'}\ldots S_j|(1 - \frac{3}{k}) - O(1) \tag{by \eqref{PB: rmk: pb not pdfs2: EQ5}} \\
    & = (n - |y| - |\overline{S_{j'-1}}|(1 - \frac{3}{k}) - O(1) \notag \\
    & \geq (n - |y|)(1 - \frac{4}{k}). \label{PB: rmk: pb not pdfs2: EQ6}
\end{align}
Then, noting that $n = \Omega(j^3)$ and that $|y| = O(j^2)$, by Equation \eqref{PB: rmk: pb not pdfs2: EQ6} we have that  \begin{equation}
    |C(S \upharpoonright n)| \geq n(1 - \frac{5}{k}).\label{PB: rmk: pb not pdfs2: EQ7}
\end{equation}
As $C$ was arbitrary, we therefore have that \begin{equation}
    \rho_{\textrm{UPD}}(S) > 1 - \frac{6}{k}. \label{PB: rmk: pb not pdfs2: EQ8}
\end{equation}

Next we build an ILPDC $C'$ that is able to compress prefixes of $S$. Let $C'$ be the ILPDC which outputs its input for some prefix $S_1\ldots S_i$. Then, for all $j > i$, $C'$ compresses $S_j$ as follows: On $S_j$, $C'$ outputs its input on $R_j^{|R_j|}1^k$ while trying to identify the $1^k$ flag. Once the flag is found, $C'$ pops the flag from its stack and then begins to read an $(R_j^{-1})^{|R_j|}$ zone. On $(R_j^{-1})^{|R_j|}$, $C'$ counts modulo $v$ to output a zero every $v$ bits, and uses its stack to ensure that the input is indeed $(R_j^{-1})^{|R_j|}$. If this fails, $C'$ outputs an error flag and enters an error state and from then on outputs its input. Furthermore, $v$ is cleverly chosen such that for all but finitely many $j$, $v$ divides evenly in $|R_j|$. Specifically we set $v = k^a$ for some $a \in \N$. A complete description of $C'$ is provided at the end of this proof.

Next we will compute the compression ratio of $C'$ on $S$. We let $p$ be such that for all $j \geq p$, $v$ divides evenly into $|R_j|$. $C'$ will output its input on $\overline{S_{p}}$ and begin compressing on the succeeding zones. Also, note that the compression ratio of $C'$ on $S$ is largest on prefixes ending with a flag $1^k$. Hence, consider some prefix $\overline{S_{j-1}}R_j^{|R_j|}1^k$ of $S$. We have that for $n$ sufficiently large
\begin{align}
     \frac{|C(\overline{S_{j-1}}R_j^{|R_j|}1^k)|}{|\overline{S_{j-1}}R_j^{|R_j|}1^k|} &\leq \frac{|\overline{S_{p-1}}| + \sum_{i = p}^j (|R_i|^2 + k + \frac{|R_i|^2}{v}) - \frac{|R_j|^2}{k}}{|\overline{S_{j-1}}R_j^{|R_j|}1^k|} \notag \\
     & \leq \frac{|\overline{S_{p-1}}|}{|\overline{S_{j-1}}|} +  \frac{(1 + \frac{1}{v})\sum_{i = 1}^j(kt_i)^2 + jk - \frac{(kt_j)^2}{v}}{|\overline{S_{j-1}}|} \notag  \\
     & \leq \frac{1}{6v} + \frac{(1 + \frac{1}{v})\sum_{i = 1}^j(kt_i)^2 + jk - \frac{(kt_j)^2}{v}}{ 2k^2 \sum_{i = 1}^{j-1}t_i^2} \tag{for $j$ large} \\
    & \leq \frac{1}{6v} + \frac{(1 + \frac{1}{v})\sum_{i = 1}^jt_i^2 + \frac{j}{k} - \frac{t_j^2}{v}}{ 2 \sum_{i = 1}^{j-1}t_i^2} \notag\\
    & \leq \frac{1}{6v} + \frac{(1+ \frac{1}{v})\sum_{i = 1}^{j-1}t_i^2}{2 \sum_{i = 1}^{j-1}t_i^2} + \frac{t_j^2}{2 \sum_{i = 1}^{j-1}t_i^2} + \frac{j}{2k \sum_{i = 1}^{j-1}t_i^2}  \notag \\
    & \leq \frac{1}{6v} + \frac{1}{2} + \frac{1}{2v} + \frac{3(jk)^2}{(j-1)(j)(2j+1)} + \frac{3}{k(j-1)(2j+1)} \notag \\
    & \leq \frac{1}{6v} + \frac{1}{2} + \frac{1}{2v} + \frac{1}{6v} + \frac{1}{6v} \tag{for $j$ large} \\
    & = \frac{1}{2} + \frac{1}{v}.
\end{align}
As $v$ can be chosen to be arbitrarily large, we therefore have that \begin{align}
    R_{\textrm{PD}}(S) \leq \frac{1}{2}. \label{PB: rmk: pb not pdfs2: EQ9}
\end{align}
Hence, for $n$ large, by Equations \eqref{PB: rmk: pb not pdfs2: EQ8} and \eqref{PB: rmk: pb not pdfs2: EQ9} it follows that for all $C \in \ILUPDC$ \begin{align}
    |C(S \upharpoonright n)| - |C'(S \upharpoonright n)| & \geq (1 - \frac{6}{k} - \frac{1}{k})n - (\frac{1}{2} + \frac{1}{k})n \\
    & = (\frac{1}{2} - \frac{8}{k}). \label{PB: rmk: pb not pdfs2: EQ10}
\end{align}

Hence, choosing $k$ large such that $\frac{8}{k} \leq \beta$ gives us our desired result of PD-depth$(S) \geq \frac{1}{2} - \beta.$

\begin{claim}
PB-depth$(S) \geq 1 - \beta.$ 
\end{claim}

Finally we examine the pebble depth of $S$. First we note that by Equation \eqref{PB: rmk: pb not pdfs2: EQ8}, it holds that \begin{align}
    \rho_{\textrm{FS}}(S)  > 1 - \frac{6}{k}. \label{PB: rmk: pb not pdfs2: EQ11}
\end{align}

Next consider the pebble-transducer $T$ that reads its input the following way: $T$ reads its input in chunks of size $2$ trying to find flags of uneven bits. If $T$ reads a chunk $10$ in its input, $T$ then scans right continuing to read its input in chunks of size two until it finds two unequal bits. $T$ uses the two flags and the string of the form $d(x)$ between the flags to print the string $x^{|x|}$. That is, if $T$ reads an input with the substring $10d(x)b_1b_2$, with $b_1,b_2 \in \bin$, $b_1 \neq b_2$, and  $x \in \fbins$, then $T$ outputs $x^{|x|}$ on that substring. If instead $T$ reads the chunk $01$, then $T$ reads its input in chunks of size $2$, outputting a single bit from each chunk if the bits match until it sees an unequal chunk or it reaches the end of the tape. That is, if $T$ reads an input with the substring $01d(x)b_1b_2$, with $b_1,b_2 \in \bin$, $b_1 \neq b_2$, and  $x \in \fbins$, or the tape ends with $01d(x)\vdash$, then $T$ outputs $x$. $T$ enters its final state upon seeing $\vdash$ if the last flag it saw was $01$, i.e. $T$ must `print' at least the empty string to enter a final state. A full description of $T$ is provided at the end of this proof.

Consider an arbitrary prefix $S \upharpoonright n$ of $S$. Let $j$ be such that $\overline{S_{j-1}}$ is a prefix of $S \upharpoonright n$ but $\overline{S_j}$ is not. That is, $S\upharpoonright n = \overline{S_{j-1}}\cdot y$ for some $y \sqsubset S_j$. For each $i$, let $x_i$ denote the string $$x_i = 10\cdot d(R_i)\cdot 01 \cdot d(1^k) \cdot 10 \cdot d(R_i^{-1}).$$
Hence we have that $$T(x_1\ldots x_{j-1} 01\cdot d(y)) = S \upharpoonright n.$$ Then, for all $\varepsilon > 0$, for $n$ large it follows that \begin{align}
    \frac{\PBcomp{|T|}(S \upharpoonright n)}{n} & = \frac{\sum_{i=1}^{j-1}|x_i| + 2 + 2|y|}{|\overline{S_{j-1}}| + |y|} \notag\\
    & \leq \frac{4\sum_{i=1}^{j-1}|R_i| + (6 + 2k)(j-1) + 2 + 2|y|}{|\overline{S_{j-1}}|} \notag \\
    & \leq \frac{4\sum_{i=1}^{j-1}kt_i + (6 + 2k)(j-1)  + 2 + 2|S_j|}{|\overline{S_{j-1}}|} \notag \\
    & \leq \frac{4k\sum_{i=1}^{j-1}t_i + (6+2k)(j-1) + 2 + 2k + 4(kt_j)^2}{2k^2\sum_{i=1}^{j-1}t_i^2} \notag \\
    & \leq \frac{4k^2\sum_{i=1}^{j-1}i}{2k^2\sum_{i=1}^{j-1}i^2} + \frac{(6+2k)(j-1) + 2(1+k)}{2k^2\sum_{i=1}^{j-1}i^2} + \frac{4k^3j^2}{2k^2\sum_{i=1}^{j-1}i^2} \notag \\
    & = \frac{6}{2j-1} + \frac{(6+2k)(j-1) + 2(1+k)}{(j-1)(j)(2j-1)/6} + \frac{12kj^2}{(j-1)(j)(2j-1)} \notag\\
    & \leq \varepsilon. \tag{for $j$ large}
\end{align}
Hence we have that \begin{equation}
R_{\textrm{PB}}(S) = 0. \label{PB: rmk: pb not pdfs2: EQ13}
\end{equation}

Therefore, By Equations \eqref{PB: rmk: pb not pdfs2: EQ11} and \eqref{PB: rmk: pb not pdfs2: EQ13}, for all $k$ and almost every $n$ we have that
\begin{align}
    \FScomp{k}(S \upharpoonright n) - \PBcomp{|T|}(S \upharpoonright n) \geq (1 - \frac{7}{k})n - \frac{1}{k}n =  (1 - \frac{8}{k})n\geq (1 - \beta)n. \label{PB: rmk: pb not pdfs2: EQ14}
\end{align}
That is, PB-depth$(S) \geq 1 - \beta$ as desired.

For completeness, the construction of the ILPDC $C'$ is confined to the appendix subsection \ref{cons3} the construction of the PB $T$ is confined to the appendix subsection \ref{cons4}.

\end{proof}

\section{Remarks}
In this paper we developed a variant of Bennett's logical depth based on pebble transducers, and showed that it satisfies versions of the fundamental properties of depth. Specifically, we first showed that FST-trivial and PB-incompressible sequences are not PB-deep in Theorem \ref{PB: Thm: Easy Hard}. We demonstrated a slow growth type law holds in Theorem \ref{PB: Thm: SGL}. We differentiated PB-depth from FS-depth by showing the existence of a normal PB-deep sequence in Theorem \ref{PB: Thm: Normal}. We also proved the existence of PB-deep sequences in Theorem \ref{PB: Thm: PB vs LZ} which, if they are LZ-deep, have low LZ-depth. A preliminary comparison with pushdown depth was also performed in Remark \ref{PB: rmk: pb not pdfs}. 

Currently PB-depth is defined as a mixed notion between FSTs and PBs. Ideally a non-mixed version would be developed, i.e. a depth notion of $k$-PB complexity vs $k'$-PB complexity. One obstacle is finding an analogous result to Lemma \ref{FS: Lemma: GEQ Lemma} which is used to prove the existence of FS-deep sequences. The current obstacle is that in the finite-case, given an FST where $T(x) = p$ and $T(xy)=pq$, simply switching the starting state to the state which it ends reading $T(x)$ in does not mean that $y$ is a description for $q$ since one must take into account the location of the pebbles too after reading $x$.

A full comparison with PD-depth is also not performed. Remark \ref{PB: rmk: pb not pdfs} presents a sequence which is both PB-deep and PD-deep. The construction of the sequence does not make use of the ability of pebble transducers to compute the $\pref$ function. Perhaps this is the approach to take to identify a sequence which is PB-deep but not PD-deep? Similarly, does there exist a sequence which is LZ-deep but not PB-deep? 

Furthermore, in this paper knowing that PB-incompressible sequences exist was sufficient for our desired results. Based on this, a result which could also potentially be expanded upon is Lemma \ref{PB: Lemma: pebble random} in which we showed that ML-random sequences are PB-incompressible. Just as it is known that a sequence is FS-incompressible if and only if the sequence is normal, a similar result which classifies what sequences are PB-incompressible is welcome.

  \bibliographystyle{plainurl}
\bibliography{biblio.bib}

\newpage

\section*{Appendix}
In the following we present constructions of pebble transducers described in the main body of the paper. Some transitions are omitted for succinctness, however one is free to assume for every transition not described, the transducer enters and remains in a non-final extra state, thus ensuring determinism.

\subsection{Construction from Lemma \ref{PB: Lemma: SGL No Work With PB}}
\label{cons1}
For completeness, we provide the following construction for $T$: As $T$ is a $1$-pebble transducer, the pebble placement part of the transition and output function will have value $0$ or $1$ indicating whether or not the pebble is present on the current square of the input tape.

Let $T = (Q,q_0,\{q_f\},1,\delta,\nu)$ be as follows. $T$ has the following set of states:
\begin{enumerate}
    \item $q_s$ the start state.
    \item $q_p$ is the state $T$ enters when it needs to move its pebble.
    \item $q^b$ is the state which records the first bit for $b \in \bin$ when examining a block of size $2$.
    \item $q_l$ is the state used when $T$ continuously moves its head to the left side of the tape.
    \item $q_1,q_2,q_3$ are the states used to print the prefixes of the input.
    \item $q_i$ is the state where $T$ acts as the identity transducer.
    \item $q_f$ is the final state.
\end{enumerate} 
Beginning in the start state, $T$ moves its head to the right and enters the pebble placement state
$$\delta(q_s,\dashv,0) = (q_p, +1).$$

Beginning in $q_p$, $T$ then reads the next two bits. $T$ first records the first bit and moves right $$\delta(q_p,b,0) = (q^b,+1).$$
Then reading the second bit, if it matches the first bit, $T$ places a pebble onto the square and enters the state for scanning to the left. If they do not match, $T$ moves right and enters the identity state. That is

\begin{equation*}
    \delta(q^b,a,0) = 
    \begin{cases}
        (q_l,\push) & \textrm{if $a = b$}, \\
        (q_i,+1) &\textrm{if $a \neq b$}.
    \end{cases}
\end{equation*}

In $q_l$, $T$ scans left to the end of the tape, i.e. for $b,c \in \bin$, $$\delta(q_l,b,c) = (q_l,-1).$$
When $T$ reaches the end of the tape, it begins reading in chunks of size two, printing every second bit, until it sees the square containing the pebble. $T$ first moves its head right,
$$\delta(q_l,\dashv,0) = (q_1,+1).$$ $T$ then moves its head to the right to the second square on any bit, $$\delta(q_1,b,0) = (q_2,+1).$$ In $q_2$, if the current square contains the pebble, $T$ pops the pebble and moves it forward two squares. If it does not, $T$ moves right and returns to $q_1.$ That is, on any bit $b$,

\begin{equation*}
    \delta(q_2,b,c) = 
    \begin{cases}
        (q_1,+1) & \textrm{if $c = 0$}, \\
        (q_3,\pop) &\textrm{if $c = 1$}.
    \end{cases}
\end{equation*}

In $q_3$, $T$ returns to $q_p$ and moves its head to the right to begin the process of moving the pebble again. That is,
$$\delta(q_3,b,0) = (q_p,+1).$$

When in state $q_i$, $T$ moves right regardless of the bit read. That is,
$$\delta(q_i,b,0) = (q_i,+1).$$

$T$ enters its final state if $T$ reaches the right hand side of the tape in states $q_p,q^b$ or $q_i$. That is, for $q \in \{q_p,q^b,q_i\}$, $$\delta(q,\vdash,0) = (q_f,-1).$$

$T$ outputs the empty string on all transitions except in the following cases where it prints the bit on the current square: 
\begin{align*}
    \textrm{For $c \in \bin$, }\nu(q_2,b,c) = b\, \textrm{ and }\nu(q_i,b,0) = b.
\end{align*}
This completes the construction of $T$.

\subsection{Construction  from Theorem \ref{PB: Thm: PB vs LZ}}
\label{cons2}

For completeness, the following is a construction for $T$: $T=(Q,q_o,F,1,\delta,\nu)$ is the $1$-pebble transducer whose states are are follows: 
\begin{enumerate}
    \item $q_0$ the start state,
    \item $q_{i,w}$ for $w \in \bin^{2k}$ the just printing states,
    \item $q_1$ the state used to check whether the transducer just prints or needs to print the reverse too,
    \item $q_p$ a state used to place the pebble,
    \item $q_{r,w}$ for $w \in \bin^{2k}$, the state where $T$ moves right printing but will print the reverse too,
    \item $q_{f}$ the state when scanning left along the flag before printing the reverse,
    \item $q_l$ the state used to print the reverse moving left,
    \item $q_{s,w}$ for $w \in \bin^{2k}$ used to scan right,
    \item $q_F$ the final state.
\end{enumerate}
So $F = \{q_F\}$.

From the start state, $T$ moves to state $q_{i,0^{2k}}$ and prints nothing. That is, $$\delta(q_0,\dashv,0) = (q_{i,0^{2k}},+1),$$ and $$\nu(q_0,\dashv,0) = \lambda.$$ 
From here, $T$ continuously prints what is under its head moving right until it sees the end of a flag. At the end of the flag it moves to $q_1$. That is for $w \in \bin^{2k}$ and $b \in \bin$ \begin{equation*}
    \delta(q_{i,w},b,0) = 
        \begin{cases}
            (q_{i,w[1..]b},+1) & \text{if $w \neq 1^{2k}$ or ($w = 1^{2k}$ and $b = 1$),} \\
            (q_1,+1) & \text{if $w = 1^{2k}$ and $b = 0$,}
        \end{cases}
\end{equation*}
and 
\begin{equation*}
    \nu(q_{i,w},b,0) = 
        \begin{cases}
            b & \text{if $w \neq 1^{2k}$ or ($w = 1^{2k}$ and $b = 1$),} \\
            \lambda & \text{if $w = 1^{2k}$ and $b = 0$.}
        \end{cases}
\end{equation*}

In $q_1$, $T$ has just read a $0$ after a flag of $1^{2k}$. If $T$ reads a $1$ in $q_1$, $T$ moves right and returns to $q_{i,0^{2k}}$ the initial printing state. If $T$ reads a $0$, $T$ moves right and enters the state $q_p$ and places its pebble on its tape. That is,
\begin{equation*}
    \delta(q_1,b,0) = 
        \begin{cases}
            (q_p,+1) & \text{if $b = 0$,}\\
            (q_{i,0^{2k}},+1) & \text{if $b = 1$.}
        \end{cases}
\end{equation*}
$T$ prints nothing in $q_1$. That is, for $b,c \in \bin$ $$\nu(q_1,b,c) = \lambda.$$

In $q_p$, $T$ places a pebble on its current square and enters state $q_{r,0^{2k}}$ and prints nothing. That is, for $b \in \bin$,
$$\delta(q_p,b,0) = (q_{r,0^{2k}},\push),$$ and $$\nu(q_p,b,0) = \lambda.$$

$T$ moves its head to the right printing what it reads when in states $q_{r,w}$. It does this until it sees the end of a $1^{2k}$ flag, upon which it enters state $q_f$ moving its head to the left. That is, for $b,c \in \bin$, $w \in \bin^{2k}$,

\begin{equation*}
    \delta(q_{r,w},b,c) = 
        \begin{cases}
            (q_{r,w[1..]b},+1) & \text{if $w \neq 1^{2k}$ or ($w = 1^{2k}$ and $b = 1$),} \\
            (q_f,-1) & \text{if $w = 1^{2k}$ and $b = 0$,}
        \end{cases}
\end{equation*}
and
\begin{equation*}
    \nu(q_{r,w},b,c) = 
        \begin{cases}
            b & \text{if $w \neq 1^{2k}$ or ($w = 1^{2k}$ and $b = 1$),} \\
            \lambda & \text{if $w = 1^{2k}$ and $b = 0$.}
        \end{cases}
\end{equation*}

$T$ moves its head to the left printing nothing while in $q_f$ until it sees a $0$, that is, the end of the $1^{2k}$ flag zone. When it sees a $0$, $T$ begins printing what it reads and enters state $q_l$. That is for $b \in \bin$,
\begin{equation*}
    \delta(q_f,b,0) = 
        \begin{cases}
            (q_f,-1) & \text{if $b = 1$,} \\
            (q_l,-1) & \text{if $b = 0$,}
        \end{cases}
\end{equation*}
and 
\begin{equation*}
    \nu(q_f,b,0) = 
        \begin{cases}
            \lambda & \text{if $b = 1$,} \\
            0 & \text{if $b = 0$.}
        \end{cases}
\end{equation*}

In $q_l$, $T$ moves its head to the left printing what it sees until it sees the square with the pebble. When $T$ sees the pebble, $T$ removes the pebble and enters state $q_{s,0^{2k}}$. That is for $b,c \in \bin$
\begin{equation*}
    \delta(q_l,b,c) = 
        \begin{cases}
            (q_l,-1) & \text{if $c = 0$,} \\
            (q_{s,0^{2k}},\pop) & \text{if $c = 1$,}
        \end{cases}
\end{equation*}
and 
$$\nu(q_l,b,c) = b.$$

$T$ moves its head to the right printing nothing until it sees the end of a $1^{2k}$ flag, upon which it enters state $q_1$ to begin the process of printing a new zone again. That is, for $b \in \bin,w \in \bin^{2k}$
\begin{equation*}
    \delta(q_{s,w},b,0) = 
        \begin{cases}
            (q_{s,w[1..]b},+1) & \text{if $w \neq 1^{2k}$ or ($w = 1^{2k}$ and $b = 1$),} \\
            (q_1,+1) & \text{if $w = 1^{2k}$ and $b = 0$,}
        \end{cases}
\end{equation*}
and 
$$\nu(q_{s,w},b,0) = \lambda.$$

For $w \in \bin^{2k}$, if $T$ is in state $q_{i,w}$ (the just printing states without reversing) or in state $q_1$ (where $T$ checks if the next zone is just printing or printing and reversing) and sees $\vdash$ indicating the right hand side of the tape, $T$ enters $q_F$ the final state and halts, printing nothing. That is for $w\in \bin^{2k}$

$$\delta(q_{i,w},\vdash,0) = \delta(q_1,\vdash,0) = (q_F,-1),$$
and $$\nu(q_{i,w},\vdash,0) = \nu(q_1,\vdash,0) = \lambda.$$

This completes the construction of $T$.

\subsection{Construction from Remark \ref{PB: rmk: pb not pdfs} : ILPDC}
\label{cons3}

For completeness we now present a full description of the ILPDC $C'$: Let $Q$ be the following set of states: \label{ILPDC Construction}
\begin{enumerate}
    \item the start state $q_0^s$,
    \item the counting states $q^s_1,\ldots q^s_m$ and $q_0$ that count up to $m = |\overline{S_{p-1}}|$,
    \item the flag checking states $q_1^{f_1},\ldots,q_k^{f_1}$ and $q_1^{f_0},\ldots,q_k^{f_0}$,
    \item the pop flag states $q_0^F,\ldots, q_k^F$,
    \item the compress states $q_1^c,\ldots, q_{v+1}^c$,
    \item the error state $q_e$.
\end{enumerate}

\noindent We now describe the transition function of $C'$. At first, $C'$ counts om $q_0^s$ to $q^s_m$ to ensure that for later $R_j$ zones, $v$ divides evenly into $|R_j|$. That is, for $0 \leq i \leq m-1$, $$\delta(q_i^s,x,y) = (q_{i+1}^s,y)$$ and $$\delta(q_m^s,\lambda,y) = (q_0, y).$$

\noindent Once this counting has taken place, an $R_j$ zone begins. Here, the input is pushed onto the stack and $C'$ tries to identify the flag $1^k$ by examining group of $k$ symbols. We set \begin{align*}
    \delta(q_0,x,y) = 
    \begin{cases}
         (q_1^{f_1},xy) & \textrm{if $x = 1$}\\
         (q_1^{f_0},xy) & \textrm{if $x \neq 1$}
    \end{cases}
\end{align*}
and for $1 \leq i \leq k-1$, $$\delta(q_i^{f_0},x,y) = (q_{i+1}^{f_0},xy)$$ and \begin{align*}
    \delta(q_i^{f_1},x,y) = 
    \begin{cases}
         (q_{i+1}^{f_1},xy) & \textrm{if $x = 1$} \\
         (q_{i+1}^{f_0},xy) & \textrm{if $x \neq1$.}
    \end{cases}
\end{align*}
If the flag $1^k$ is not detected after $k$ symbols, the test begins again. That is $$\delta(q_k^{f_0},\lambda,y) = (q_0,y).$$
If the flag is detected, the pop flag state is entered. $\delta(q_k^{f_1},\lambda,y) = (q_0^F,y).$ The flag is then removed from the stack, that is, for $0 \leq i \leq k$ $$\delta(q_i^F,\lambda,y) = (q_{i+1}^F, \lambda)$$ and $$\delta(q_k^F,\lambda,y) = (q_1^c,y).$$ 

\noindent $C'$ then checks using the stack, that the next part of the input it reads is $R_j^{-1}$, counting modulo $v$. If the checking fails, the error state is entered. That is for $1 \leq i \leq v$, \begin{align*}
    \delta(q_i^c,x,y) = 
    \begin{cases}
         (q_{i+1}^c,\lambda) & \textrm{if $x  = y$}\\
         (q_e, y) & \textrm{if $x \neq y$ and $y \neq z_0$} \\
         (q_1^{f_1},xz_0) & \textrm{if $x=1$ and $y = z_0$}\\
         (q_1^{f_0},xz_0) & \textrm{if $x\neq1$ and $y = z_0$}.
    \end{cases}
\end{align*}
Once $v$ symbols are checked, the checking starts again. That is $$\delta(q_{v+1}^c,\lambda,y) = (q_1^c,y).$$
The error state is the loop $$\delta(q_e,x,y) = (q_e,y).$$

We now describe the output function of $C'$. Firstly, on the counting states, $C'$ outputs its input. That is, for $0 \leq i \leq m-1$ $$\nu(q_i^s,x,y) = x.$$
On the flag checking states $C'$ outputs its input. That is, for $1 \leq i \leq k-1$ $$\nu(q_i^{f_0},x,y) = \nu(q_i^{f_1},x,y) = x.$$
$C'$ outputs nothing while in the flag popping states $q_0^F,\ldots, q_k^F$ and on the compression states $q_1^c,\ldots,q_{v+1}^c$ except in the case when $v$ symbols have just been checked. That is, $$\nu(q_{v}^c,x,y) = 0 \textrm{ if $x = y$}.$$
When an error is seen, a flag is outputted. That is for $1 \leq i \leq v$ $$\nu(q_i^c,x,y) = 1^{3m + i}0x \textrm{ if $x \neq y$ and $y \neq z_0$}.$$
$C'$ outputs its input while in the error state. That is, $$\nu(q_e,x,y) = x.$$

Lastly we verify that $C'$ is in fact IL. If the final state is not an error state, then all $R_j$ zones and $1^k$ flags are output as in the input. If the final state is $q_i^c$ then the number $t$ of zeros after the last flag in the output along with $q_i^c$ determines that the last $R_j^{-1}$ zone read is $tv + i-1$ bits long.
If the final state is $q_e$, then the output is of the form $$aR_j1^k0^t1^{3m + i}0b$$ for $a,b \in \fbins.$ The input is uniquely determined to be the input corresponding to the output $aR_j1^k0^t$ with final state $q_1^c$
 followed by $$R_j^{-1}[tv..tv + (i-1)-1].$$ As $1^{3m}$ does occur anywhere as a substring of $S$ post the prefix $\overline{S_{p-1}}$, its first occurrence post $\overline{S_{p-1}}$ as part of an output must correspond to an error flag.

\subsection{Construction from Remark \ref{PB: rmk: pb not pdfs} : PB}
\label{cons4}

For completeness, the following is a full description of the pebble transducer $T$: Let $Q$ be the following set of states of $T$:
\begin{enumerate}
    \item the start state $q_s$,
    \item the accepting state $q_a$,
    \item the failure state $q_d$, 
    \item the initial flag identifying states $q^i_s,q_s^0$ and $q_s^1$,
    \item the `just print' states $q_p,q_p^0$ and $q_p^1$,
    \item the `place pebble' states $q_r^,q_r^0,q_r^1$,
    \item the states used to find a flag when scanning left $q_l,q_l^0$ and $q_l^1$,
    \item the `print square' states $q_i^{-1},q_i, q_i^0$ and $q_i^1$,
    \item the states used to find and pop the pebble from the tape $q_f$ and $q_f'$.
\end{enumerate}

We first describe the transition function $\delta$ of $T$. Beginning in the start state, $T$ checks whether the next two bits contain $01$ or $10$ to indicate whether it is entering a \emph{print} or \emph{print-square} zone respectively.  $T$ first moves right off of $\dashv$$$\delta(q_s,\dashv,0) = (q^i_s,+1).$$ In $q^i$, $T$ reads what is under its head and moves right to check the next bit. That is, \begin{equation*}
    \delta(q^i_s,b,c) = 
    \begin{cases}
        (q_d,+1) & \textrm{if $b =\, \dashv$,} \\
        (q_d,-1) & \textrm{if $b =\, \vdash$,} \\
        (q^0_s,+1) & \textrm{if $b = 0$,}\\
        (q^1_s,+1) & \textrm{if $b = 1$}.
    \end{cases}
\end{equation*}
$T$ then checks if the next bit is different from the previous bit, i.e. if a flag has just been read. If they are the same or the end of the tape has been read, $T$ enters the failure state. That is, \begin{align*}
    \delta(q^b_s,b',c) = 
    \begin{cases}
        (q_d,-1) & \textrm{if $b = \, \vdash$ or $b' = b$,} \\
        (q_p,+1) & \textrm{if $bb' = 01$,} \\
        (q_r,+1) & \textrm{if $bb' = 10$.}
    \end{cases}
\end{align*}

If the flag read was $01$, $T$ enters the `just print' states beginning with state $q_p$. Here, $T$ reads its input in chunks of size two. $T$ scans left until it sees a chunk of two unmatching bits, that is, another flag and enters the appropriate state. If it reaches the right end of the tape, it enters the final state. That is, beginning in state $q_p$, $T$ reads the first bit of a chunk \begin{align*}
    \delta(q_p,b,c) =\begin{cases}
        (q_p^b,+1) & \textrm{if $b \in \bin$,}\\
        (q_a, -1) & \textrm{if $b = \,\vdash$}.
    \end{cases}
\end{align*}
Then in state $q_p^b$, if the next bit read matches $b$, $T$ enters state $q_p$ again, otherwise it knows it has just read a flag. That is \begin{align*}
    \delta(q_p^b,b',c) = \begin{cases}
        (q_p,+1) & \textrm{if $b = b'$,} \\
        (q_p, + 1) & \textrm{if $bb' = 01$,} \\
        (q_r,+1) & \textrm{if $bb' = 10$,}\\
        (q_a, -1) & \textrm{if $b = \,\vdash$}.
    \end{cases}
\end{align*}

If $T$ has read the flag $10$, it enters the `print square' zone. $T$ must first place its pebble on its tape. Starting in state $q_r$, $T$ reads its input and then moves to the right checking if the two bits it has just read match. If they match, $T$ places its pebble on the tape, otherwise it knows it has just read another flag. That is \begin{align*}
    \delta(q_r,b,0) = \begin{cases}
    (q_d,-1) & \textrm{if $b = \, \vdash$,}\\
    (q_r^0,+1) & \textrm{if $b = 0$,} \\
    (q_r^1,+1) & \textrm{if $b = 1$,}
    \end{cases}
\end{align*}
and
\begin{align*}
    \delta(q_r^b,b',0) = \begin{cases}
        (q_d,-1) & \textrm{if $b' = \, \vdash$,}\\
        (q_l,\push) & \textrm{if $b = b'$,}\\
        (q_r,+1) & \textrm{if $bb' = 10$,}\\
        (q_p,+1) & \textrm{if $bb' = 01$}.
    \end{cases}
\end{align*}

Once the pebble is placed, beginning in state $q_l$, $T$ scans left while reading in chunks of size two to find the last $10$ flag it has read. That is, $$\delta(q_l,b,c) = (q_l^b,-1)$$ and \begin{align*}
    \delta(q_l^b, b', c) = \begin{cases}
        (q_l, -1) & \textrm{if $b = b'$,} \\
        (q_i^{-1},+1) & \textrm{if $b'b = 10$,}\\
        (q_d, + 1) & \textrm{otherwise}.
    \end{cases}
\end{align*}
Once the $10$ flag is found, beginning in state $q_i^{-1}$, $T$ moves to the right $$\delta(q_i^{-1},b,c) = (q_i,+1).$$ Using states $q_i,q_i^0$ and $q_i^1$, $T$ scans right reading in chunks of size two trying to find the next flag. That is \begin{align*}
    \delta(q_i,b,c) = \begin{cases}
        (q_d,-1) & \textrm{if $b = \, \vdash$,} \\
        (q_i^b, + 1) & \textrm{if $b \in \bin$,}
    \end{cases}
\end{align*} and
\begin{align*}
    \delta(q_i^b,b',c) = 
    \begin{cases}
        (q_d,-1) & \textrm{if $b = \, \vdash$,} \\
        (q_i,+1) & \textrm{if $b = b'$,}\\
        (q_f, - 1) & \textrm{if $b \neq b'$}.
    \end{cases}
\end{align*}
In state $q_f$, $T$ has just read a flag. $T$ then scans left to find its pebble on its tape to pop it. That is, \begin{align*}
    \delta(q_f,b,c) = 
    \begin{cases}
        (q_d,+1) & \textrm{if $b = \, \vdash$,}\\
        (q_f,-1) & \textrm{if $c = 0$,}\\
        (q_f',\pop) & \textrm{if $c = 1$}.
    \end{cases}
\end{align*} 
In state $q_f'$, $T$ moves right and re-enters state $q_r$ to place a pebble on its tape. That is, $$\delta(q_f,b,c) = (q_r, +1).$$
In the failure state $q_d$, $T$ enters a loop and so never enters $q_a$. That is \begin{align*}\delta(q_d,b,c) = \begin{cases}
    (q_d,-1) & \textrm{if $b = \, \vdash$,}\\
    (q_d,+1) & \textrm{otherwise}.
\end{cases}
\end{align*}
$T$ outputs nothing on all transitions except in the following two cases:
\begin{itemize}
    \item $\nu(p_b,b,c) = b$ (when in a `just print' state and it sees an equal chunk)
    \item $\nu(q_i^b,b,c) = b$ (when in a `print square' zone and it sees an equal block)
\end{itemize}
This completes the construction of $T$.

\end{document}